\documentclass[sigconf]{acmart}

\copyrightyear{2022} 
\acmYear{2022} 
\setcopyright{acmcopyright}
\acmConference[SIGIR '22] {Proceedings of the 45th International ACM SIGIR Conference on Research and Development in Information Retrieval}{July 11--15, 2022}{Madrid, Spain.}
\acmBooktitle{Proceedings of the 45th International ACM SIGIR Conference on Research and Development in Information Retrieval (SIGIR '22), July 11--15, 2022, Madrid, Spain}

\acmPrice{15.00}
\acmISBN{978-1-4503-8732-3/22/07} 
\acmDOI{10.1145/3477495.3531985}

\usepackage{amssymb}
\usepackage{algorithm}
\usepackage{algorithmic}
\usepackage{bibentry}
\usepackage{amsfonts}
\usepackage{subfigure}
\usepackage{multirow}
\usepackage{amsmath}
\usepackage{graphicx}
\usepackage{epsfig}
\usepackage{color}
\usepackage{epstopdf}
\usepackage{rotating}
\usepackage{caption}
\usepackage{float}
\usepackage{multicol}
\usepackage{amssymb}
\usepackage{graphicx}
\usepackage{bbding}
\usepackage{algorithm}
\usepackage{algorithmic}
\usepackage{balance}
\usepackage{microtype}
\usepackage{graphicx}
\usepackage{subfigure}
\usepackage{booktabs} 
\usepackage{multirow}
\usepackage{amsmath}

\usepackage{mathtools}
\usepackage{etoolbox}
\usepackage{cases}
\usepackage{enumitem}
\usepackage{xcolor}
\usepackage{subfigure}
\usepackage{color}
\definecolor{brown}{RGB}{139,64,0}

\usepackage{amssymb}

\usepackage{cases}
 
\usepackage{mathtools}

\usepackage{amssymb}
\usepackage{amsthm}
\usepackage{color}
\usepackage{multirow}

\newcommand{\ourname}{{GTN}}

\newcommand{\cut}[1]{{}}

\newcommand{\tA}{{\tilde{\vA}}}
\newcommand{\tD}{{\tilde{\vD}}}
\newcommand{\tL}{{\tilde{\vL}}}
\newcommand{\hA}{{\hat{\vA}}}
\newcommand{\hD}{{\hat{\vD}}}

\newcommand{\tDelta}{{\tilde{\Delta}}}

\newcommand{\ve}{{\mathbf{e}}}

\newcommand{\vA}{{\mathbf{A}}}

\newcommand{\vD}{{\mathbf{D}}}
\newcommand{\vE}{{\mathbf{E}}}
\newcommand{\vF}{{\mathbf{F}}}

\newcommand{\vI}{{\mathbf{I}}}

\newcommand{\vL}{{\mathbf{L}}}

\newcommand{\vW}{{\mathbf{W}}}
\newcommand{\vX}{{\mathbf{X}}}
\newcommand{\vY}{{\mathbf{Y}}}
\newcommand{\vZ}{{\mathbf{Z}}}

\newcommand{\cE}{{\mathcal{E}}}

\newcommand{\cG}{{\mathcal{G}}}

\newcommand{\cL}{{\mathcal{L}}}

\newcommand{\cN}{{\mathcal{N}}}

\newcommand{\cV}{{\mathcal{V}}}

\newcommand{\RR}{\mathbb{R}}

\newcommand{\sign}{\mathrm{sign}}

\newcommand{\st}{{\text{s.t.}}} 
\newcommand{\tr}{{\mathrm{tr}}} 

\newcommand{\prox}{\mathbf{prox}}

\makeatletter
\let\@@span\span
\def\sp@n{\@@span\omit\advance\@multicnt\m@ne}
\makeatother

\DeclareMathOperator*{\argmin}{arg\,min}

\newcommand{\bc}{\begin{center}}
\newcommand{\ec}{\end{center}}

\newcommand{\bdm}{\begin{displaymath}}
\newcommand{\edm}{\end{displaymath}}

\newcommand{\beq}{\begin{equation}}
\newcommand{\eeq}{\end{equation}}

\newcommand{\bfl}{\begin{flushleft}}
\newcommand{\efl}{\end{flushleft}}

\newcommand{\bt}{\begin{tabbing}}
\newcommand{\et}{\end{tabbing}}

\newcommand{\beqn}{\begin{align}}
\newcommand{\eeqn}{\end{align}}

\newcommand{\beqs}{\begin{align*}} 
\newcommand{\eeqs}{\end{align*}}  

\newtheorem{theorem}{Theorem}

\begin{document}

\fancyhead{}

\title{Graph Trend Filtering Networks for Recommendation}

\author{Wenqi Fan$^{1\dag}$, Xiaorui Liu$^{2\dag}$, Wei Jin$^{2}$, Xiangyu Zhao$^{3*}$, Jiliang Tang$^2$,  and Qing Li$^1$}

\affiliation{\institution{$^1$The Hong Kong Polytechnic University,$^2$Michigan State University,$^3$City University of Hong Kong\country{}}}

\email{wenqifan03@gmail.com,{xiaorui,jinwei2,tangjili}@msu.edu,xianzhao@cityu.edu.hk, csqli@comp.polyu.edu.hk}

\thanks{* Corresponding author. }
\thanks{$\dag$ Equal contributions..}








\begin{abstract}

Recommender systems aim to provide personalized services to users and are playing an increasingly important role in our daily lives. The key of recommender systems is to predict  how likely users will interact with items based on their historical online behaviors, e.g., clicks, add-to-cart, purchases, etc. To exploit these user-item interactions, there are increasing efforts on considering the user-item interactions as a user-item bipartite graph and then performing information propagation in the graph via  Graph Neural Networks (GNNs). Given the power of GNNs in graph representation learning, these GNNs-based recommendation methods have remarkably boosted the recommendation performance.  Despite their success, most  existing GNNs-based recommender systems overlook the existence of interactions caused by unreliable behaviors (e.g., random/bait clicks) and uniformly treat all the interactions, which can lead to sub-optimal and unstable performance. In this paper, we investigate the drawbacks (e.g., non-adaptive propagation and non-robustness) of existing GNN-based recommendation methods.  To address these drawbacks,  we introduce a principled graph trend collaborative filtering method and propose the Graph Trend Filtering Networks for recommendations (GTN) that can capture the adaptive reliability of the interactions. Comprehensive experiments and ablation studies are presented to verify and understand the effectiveness of the proposed framework.  Our implementation based on PyTorch is available\footnote{\url{https://github.com/wenqifan03/GTN-SIGIR2022}}.

\end{abstract}

 \settopmatter{printfolios=true}

%
%



\maketitle
\thispagestyle{empty}


 \section{Introduction}

Personalized recommendations aim to alleviate information overload problem through assisting users in discovering items of interest, and they have been deployed to many user-oriented online services such as E-commerce (e.g., Amazon, Taobao), and Social Media sites (e.g., Facebook, Twitter, Weibo)~\cite{chen2017attentive,fan2020graph, He2017NeuralCF,zhao2021autoloss,chen2022automated}. The key of recommender systems is to predict whether users  are likely to interact with items based on the historical interactions~\cite{He2020LightGCNSA,fan2021attacking,chen2022autogsr,zhao2021autoemb}, including clicks, add-to-cart, purchases, etc. As a widely used solution, collaborative filtering (CF) techniques are developed to model historical user-item interactions, assuming that users who behave similarly are likely to have similar preferences towards items~\cite{chen2017attentive,fan2019graph}.  

Among various collaborative filtering techniques, matrix factorization (MF) is the most popular one~\cite{He2017NeuralCF}. It embeds users and items into the vectorized representations (i.e., embeddings) and models the user-item interaction via the inner products between the user and item representations~\cite{Rendle2009BPRBP,fan2019deep_daso}. Later on, neural collaborative filtering (NeuCF) model is proposed to replace the inner product in MF model with neural networks in order to model the non-linear interactions between users and items~\cite{He2017NeuralCF}. However, in these models, the user-item interactions are only considered in the objective function for model training, which leads to unsatisfactory exploitation of the interaction data. To better exploit the user-item interactions as well as the high-order connectivity therein, graph collaborative filtering models such as NGCF~\cite{Wang2019NeuralGC} and LightGCN~\cite{He2020LightGCNSA} propose to explicitly propagate the user embedding $\ve_u^k$ and item embedding $\ve_i^k$ according to the user-item interaction through the propagation:  
\begin{align}
\label{eq:gcn_prop1}
    \ve_u^{k+1} &= \frac{1}{\sqrt{|\cN(u)|}} \sum_{i\in\cN(u)} \frac{1}{ \sqrt{|\cN(i)|} } \ve_i^k, \\
\label{eq:gcn_prop2}
    \ve_i^{k+1} &= \frac{1}{\sqrt{|\cN(i)|}} \sum_{u\in\cN(i)} \frac{1}{ \sqrt{|\cN(u)|} } \ve_u^k
\end{align}
where $\cN(u)$ denotes the set of items that user $u$ interacts with and $\cN(i)$ denotes the set of users who interact with item $i$. Such embedding propagation motivated by the feature aggregation in graph neural networks (GNNs)~\cite{ma2020unified,jin2021node,liu2021graph,wu2020comprehensive,derr2020epidemic,liu2021elastic,Fan2021JointlyAG} is the key for graph collaborative filtering models to significantly advance the state-of-art of recommender systems.

Despite their success, prior manner of modeling user-item relationships is insufficient to discover the heterogeneous reliability of interactions among instances in recommender systems. The key reason is that most existing deep recommender systems uniformly treat all the interactions. For instance, as shown in the propagation rule~\eqref{eq:gcn_prop1}, the item embedding $\ve_i^k$ is propagated to all relevant users equally.
This leads to sub-optimal representation learning for recommendations based on the historical behavior data since the reliability of the interactions is often heterogeneous.

In fact, in most e-commerce platforms, a large portion of clicks  (e.g., \emph{random/bait clicks}) are not directly related to purchases and many purchases end up with negative reviews, where users' implicit feedback (e.g., click and purchase) is unreliable and cannot reflect the actual satisfaction and preferences of users~\cite{wang2019sequential,wu2016collaborative,wang2021clicks,liu2021trustworthy}. 
As  shown in Figure~\ref{fig:intro},  with financial incentive, some sellers can sell their items with attractive exposure features (e.g., headline or cover of items), so as to deliberately mislead buyers (e.g., user 3) to interact (e.g., clicks, add-to-cart, purchases, etc.), which is known as click-bait issue; a user (i.e., user 2), who is interested in electronic products,  purchased a one-time item (i.e., handbag) for his mother's birthday present.
Worse still, such unreliable feedback may dominate the entire training dataset, and is hard to identify or prune, which may hinder the recommender systems from learning the actual user/item representations based on the user-item graph and make the recommender systems vulnerable to unreliable interactions.

\begin{figure}[tbp]
\centering
{\includegraphics[width=0.700\linewidth]{{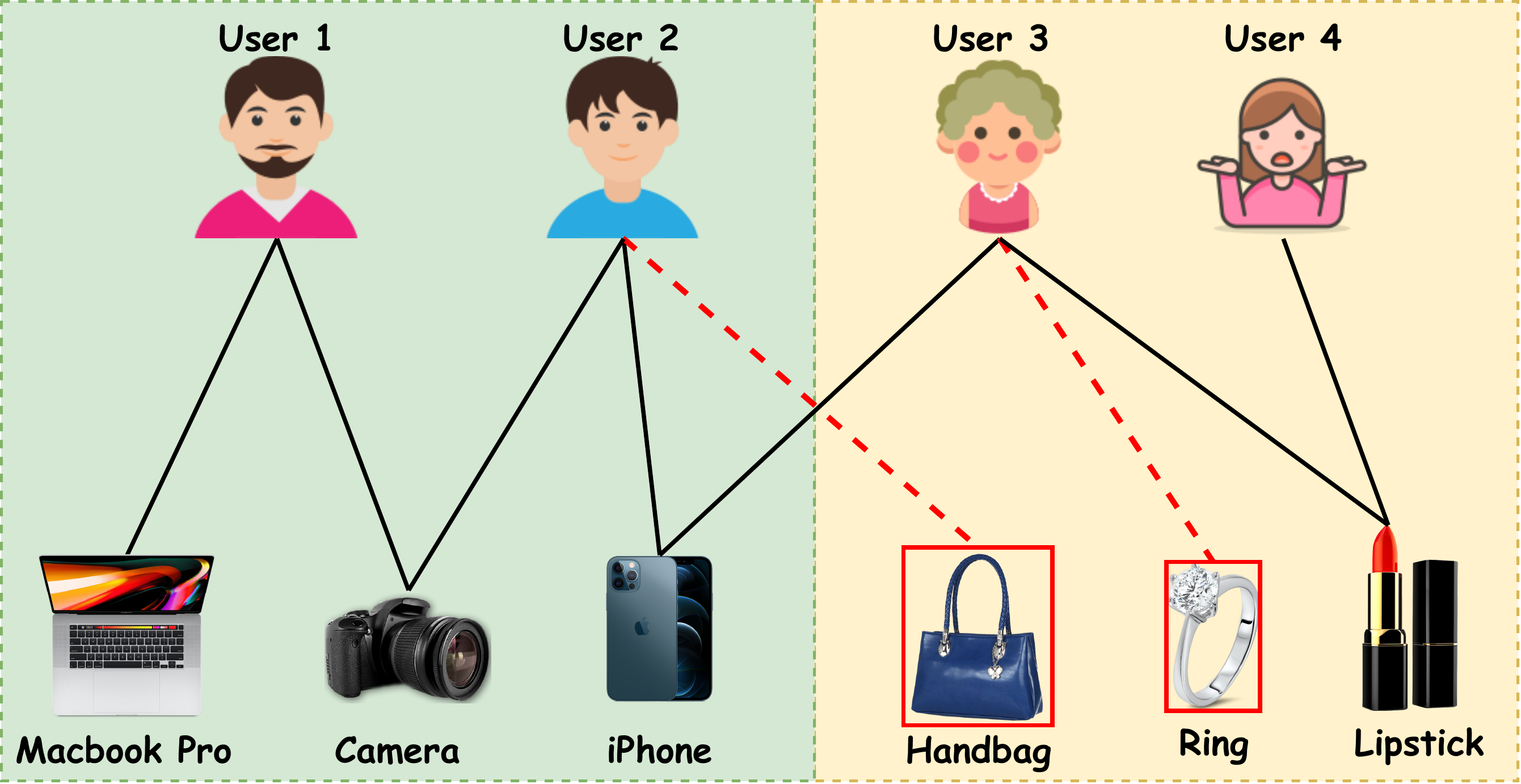}}}
\caption{An illustration on unreliable user-item interactions. User 2 bought a one-time item (i.e., Handbag) for his mother's birthday present;  User 3 was affected by the  click-bait issue and was `cheated' to interact with an item (i.e., Ring)  by the attractive  exposure features (e.g.,  headline or cover of the item).}\label{fig:intro}
\end{figure}

The aforementioned deficiency significantly limits the applications of existing collaborative filtering methods in real-world recommender systems. Therefore, it is desired to design a new collaborative filtering method that adaptively propagates the embedding
in recommender systems, which can lead to more accurate and robust recommendations.
In this paper, inspired by the concepts of trend filtering~\cite{kim2009ell_1, tibshirani2014adaptive} and graph trend filtering~\cite{wang2016trend,liu2021elastic}, 
we propose \textbf{G}raph \textbf{T}rend  filtering \textbf{N}etworks (\textbf{GTN}) to capture and learn the adaptive importance of the interactions in recommender systems. 
Our major contributions can be summarized as follows:
\begin{itemize}[leftmargin=*] 
\item 
We investigate the drawbacks of existing GNNs-based recommendation methods such as non-adaptive propagation and non-robustness from the perspective of Laplacian smoothness; 

\item 

We introduce a principled graph trend collaborative filtering technique and propose a novel graph trend filtering networks framework (GTN) to capture the adaptive reliability of the interactions between users and items in recommender systems.

\item 

Comprehensive experiments and ablation study on various real-world datasets are conducted to demonstrate the superiority of the proposed framework.
\end{itemize}


\section{Preliminary}

In this section, we introduce the notations used in this paper and briefly review some preliminaries about embedding propagation in GNNs-based collaborative filtering.

\subsection{Notations and Definitions}

We use bold upper-case letters such as $\vE$ to denote matrices and bold lower-case letters such as $\ve$ to define vectors. 
Given a matrix $\vE\in \mathbb{R}^{(n+m)\times d}$, we use $\vE_i$ to denote its $i$-th row and $\vE_{ij}$ to denote its element in $i$-th row and $j$-th column. We define the Frobenius norm and $\ell_1$ norm of matrix $\vE$ as $\|\vE\|_F=\sqrt{\sum_{ij} \vE_{ij}^2}$ and $\|\vE\|_1=\sum_{ij} |\vE_{ij}|$, respectively. We define $\|\vE\|_2=\sigma_{\max}(\vE)$ where $\sigma_{\max}(\vE)$ is the largest singular value of $\vE$.
Given two matrices $\vE, \vF \in \mathbb{R}^{n\times d}$, we define their inner product as $\langle \vE, \vF \rangle = \tr(\vE^{\top}\vF)$.

The historical interactions between users and items in recommender systems can be naturally represented as a bipartite graph $\mathcal{G}=\{\cV, \cE\}$, where the node set $\cV$ includes user nodes $\{v_1, \dots, v_n\}$ and item nodes $\{v_{n+1}, \dots, v_{n+m} \}$, and the implicit interactions can be denoted as the edges between user nodes and item nodes in the graph.
We denote the undirected edge set as $\mathcal{E}=\{e_1, \dots, e_{|\cE|} \}$. 

The graph structure of $\mathcal{G}$ can be represented as an adjacent matrix $\vA\in \mathbb{R}^{(n+m)\times (n+m)}$, where $\vA_{ij}=1$ when there exists an edge or interaction between nodes $v_i$ and $v_j$.
Note that there is no direct interaction within users or items so that the graph is a bipartite graph.  
We use $\cN(v_i)$ to denote the neighboring nodes of node $v_i$, including $v_i$ itself. 
The graph Laplacian matrix is defined as $\vL=\vD-\vA$, where $\vD$ is the diagonal degree matrix. 
Let $\Delta \in \{-1, 0, 1\}^{|\cE|\times |\cV|}$ be the oriented incident matrix, which contains one row for each edge. 
If $e_{\ell}=(i,j)$, 
then $\Delta$ has the $\ell$-th row as:
$$\Delta_{\ell} = (0, \dots, \underbrace{-1}_{i}, \dots, \underbrace{1}_{j}, \dots, 0),
$$
where the edge orientation can be arbitrary. 
In addition, we denote the embeddings of nodes in graph $\cG$ as:
\begin{gather}\label{equ:e-0}
	\vE = [\underbrace{\ve_1,\cdots,\ve_n}_{\text{users embeddings}}, \underbrace{\mathbf{e}_{n+1},\cdots,\ve_{n+m}}_{\text{item embeddings}}]^\top \in \RR^{(n+m)\times d}
\end{gather}
where $d$ is the dimension of embeddings.

\begin{figure*}[htbp]
\centering
{\subfigure[Gowalla-Recall@20]
{\includegraphics[width=0.21485\linewidth]{{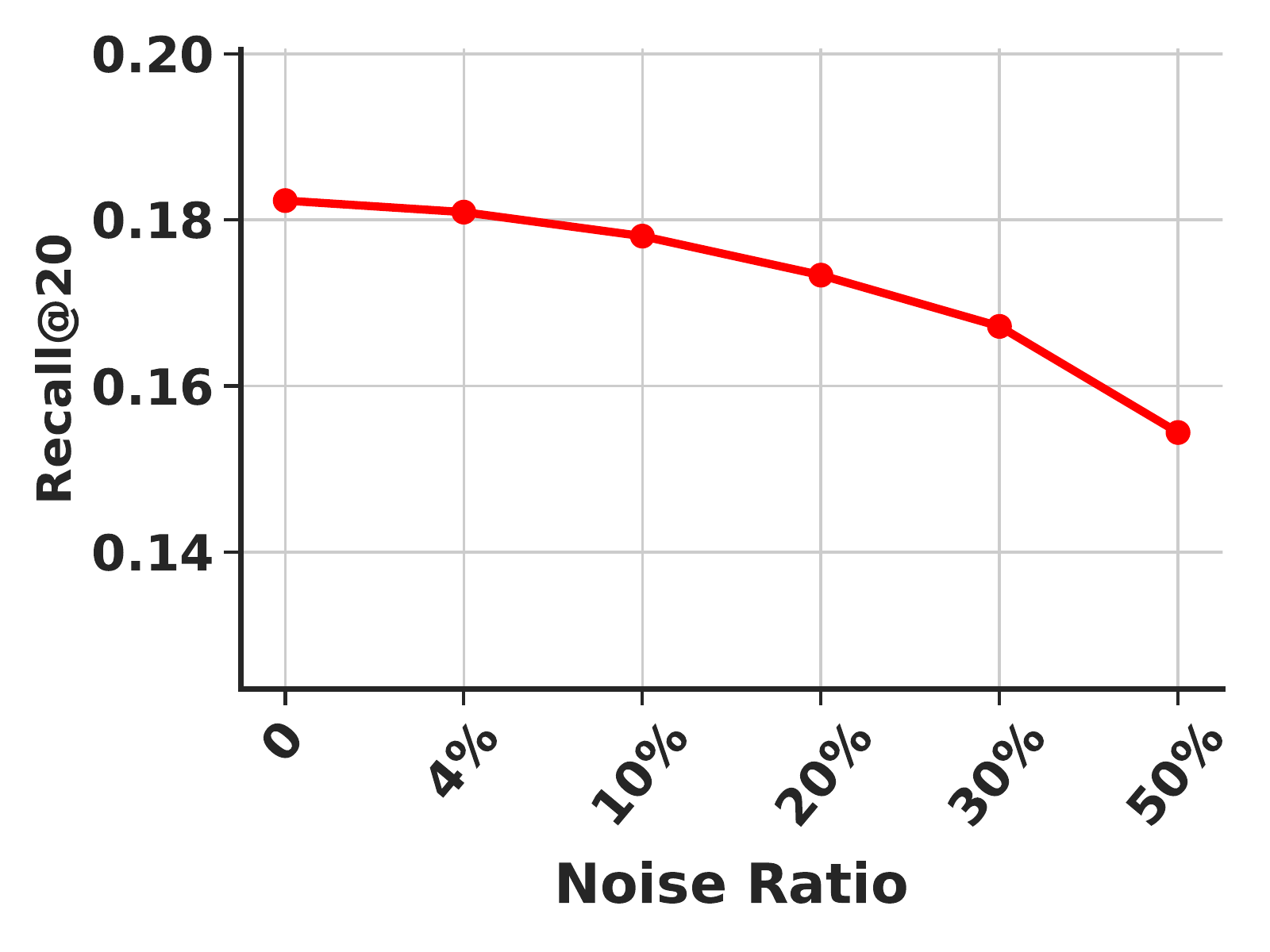}}}}
{\subfigure[Yelp2018-Recall@20]
{\includegraphics[width=0.21485\linewidth]{{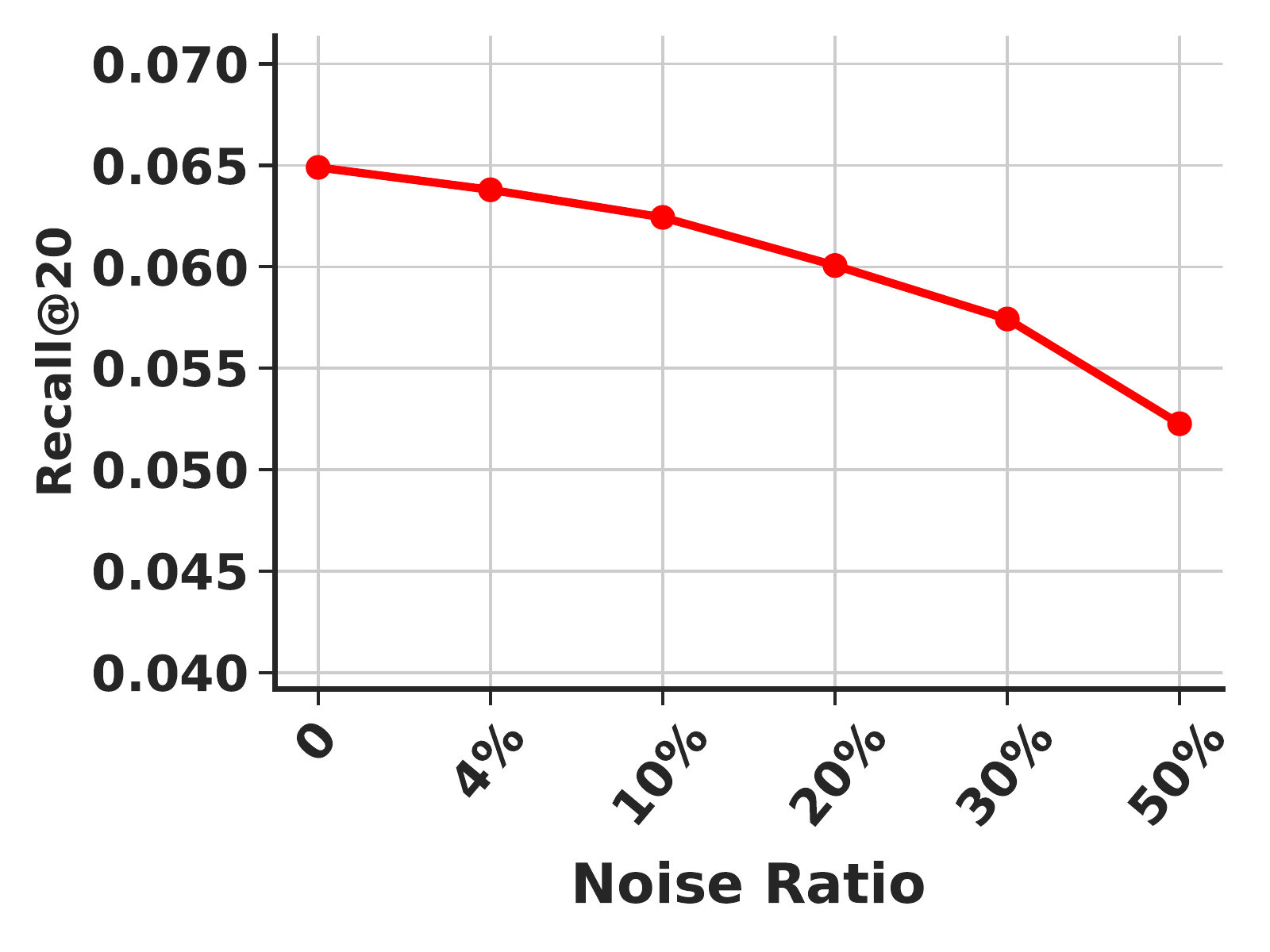}}}}
{\subfigure[Gowalla-NDCG@20]
{\includegraphics[width=0.21485\linewidth]{{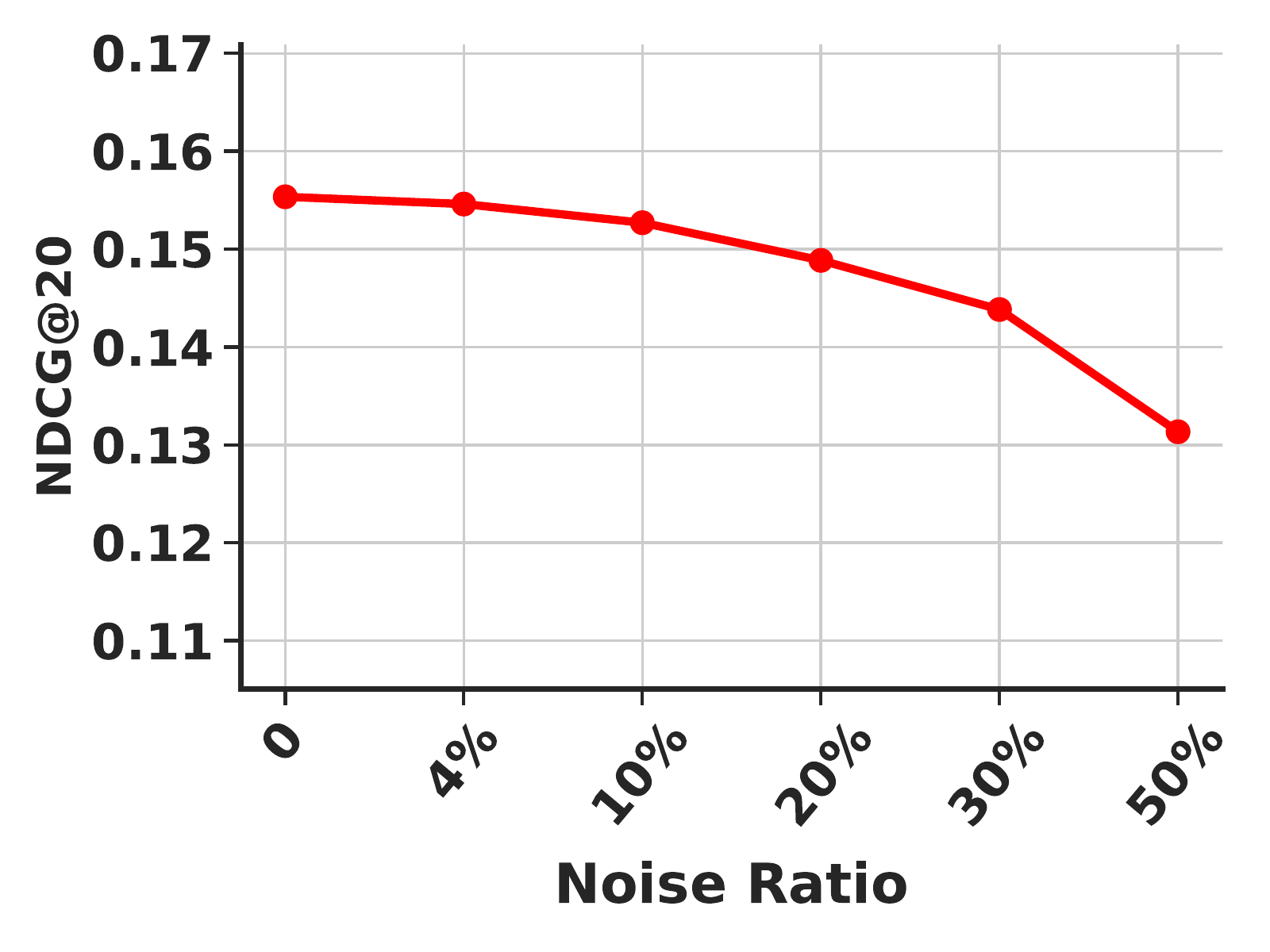}}}}
{\subfigure[Yelp2018-NDCG@20]
{\includegraphics[width=0.21485\linewidth]{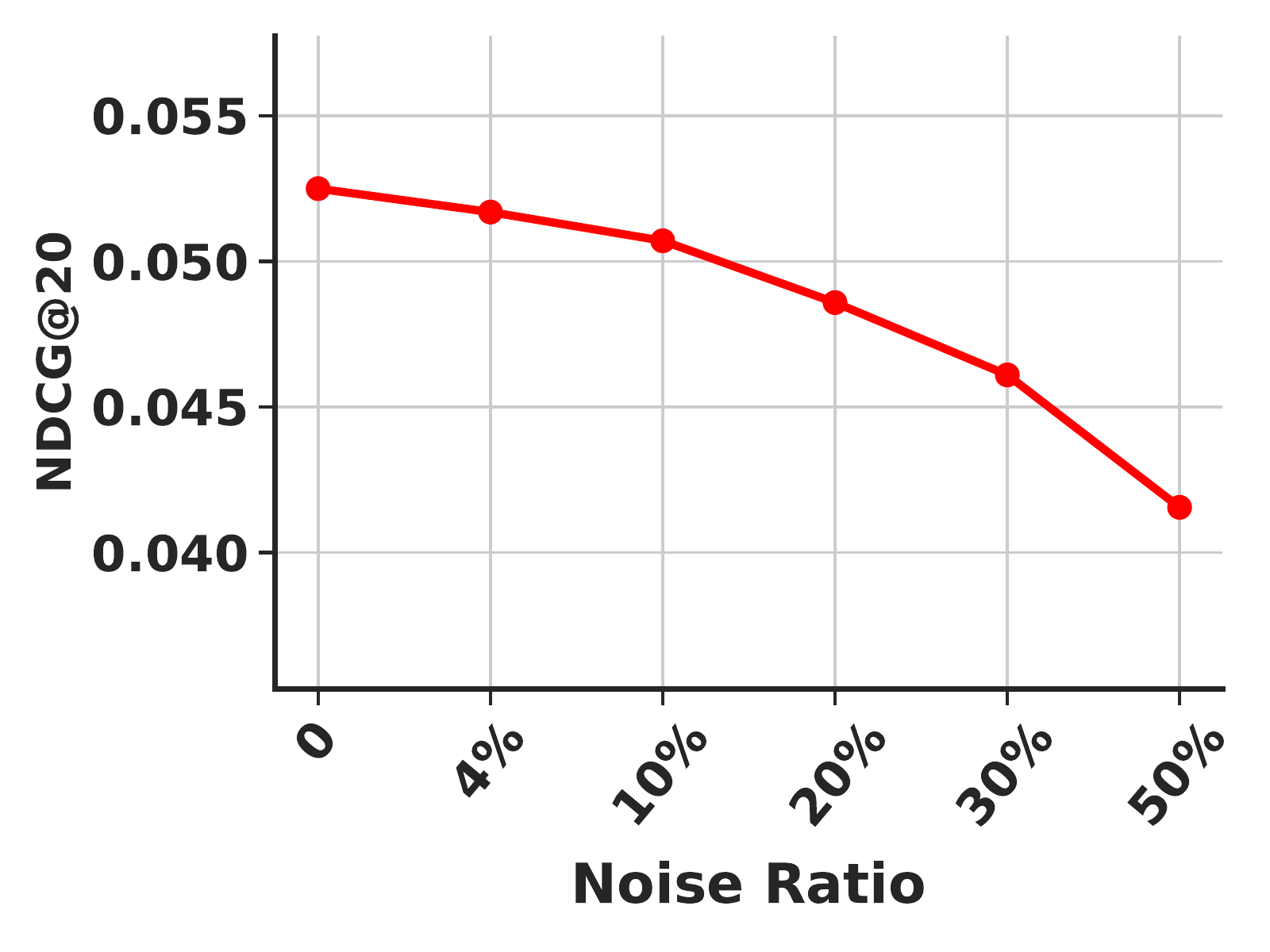}}}
\caption{Performance of LightGCN under different interaction perturbation rates.}\label{fig:perturbation_preliminary}

\end{figure*}

\subsection{Embedding Propagation and Smoothness}
\label{sec:emb_prop}
 
Given the embeddings $\vE^k$, the embedding propagation (i.e, Eq.~\eqref{eq:gcn_prop1} and Eq.~\eqref{eq:gcn_prop2}) in GNNs-based collaborative filtering can be written in a matrix form:
\begin{align}
    \vE^{k+1} = \tA \vE^k
\label{eq:prop_matrix}
\end{align}
where $\tA = \hD^{-\frac{1}{2}} \hA \hD^{-\frac{1}{2}}$ and $\hA=\vA+\vI$ whose degree matrix is $\hD$. 
 
The propagation is motivated by the graph convolutional operator in graph neural networks~\cite{kipf2016semi}, and it is shown as solving the Laplacian smoothing problem~\cite{ma2020unified} 
to enforce embedding smoothness over the user-item interactions
\begin{align}
     \argmin_{\vE\in\RR^{(n+m)\times d}} \tr(\vE^\top (\vI-\tA) \vE).
    \label{eq:laplacian_smooth1}
\end{align}

In other words, the embedding propagation in Eq~\eqref{eq:prop_matrix} is equivalent to the gradient descent on problem~\eqref{eq:laplacian_smooth1} with the initialization $\vE^k$ and stepsize $1/2$. 
It can be rewritten in an edge-wise form as: 
\begin{align}
    \tr (\vE^\top (\vI-\tA) \vE) = \sum_{(i,j)\in\cE} \| \frac{\ve_i}{\sqrt{d_i+1}} - \frac{\ve_j}{\sqrt{d_j+1}} \|_2^2
\label{eq:laplacian_smooth2}
\end{align}
with $d_i$ being the node degree of node $v_i$. It is clear that the embedding propagation essentially enforces the user embedding $\ve_i$ and item embedding $\ve_j$ to be close (i.e., embedding smoothness), if there exist interactions between them. 
As a result, their inner products become large, which  can provide positive predictions on the recommendation. However, the embedding smoothness is equally imposed over all the interactions as showed in Eq.~\eqref{eq:laplacian_smooth2}. Consequently, the embedding propagation might cause inappropriate smoothness caused by unreliable interactions such as unintended clicks. For instance, if a user randomly clicks some items, it doesn't necessarily mean the user actually likes the items. Unfortunately, the embedding propagation due to these interactions will enforce the user embedding to be similar to those items. As a result, the recommender system will rank these items at the top for the user's recommendation.
Such non-adaptive propagation results in the non-robustness of graph collaborative filtering, which limits the application in real-world recommender systems where unreliable interactions widely exist. 

To verify the non-robustness of such embedding propagation, we conduct a preliminary study to assess the performance of a representative graph collaborative filtering method when the interactions between users and items contain different ratios of random noise. This is evaluated by injecting random user-item interactions into the datasets.
The experimental results are shown in Figure~\ref{fig:perturbation_preliminary}. It can be observed that the performance of  LightGCN~\cite{He2020LightGCNSA} drops significantly when the interactions between users and items become more noisy, which confirms our discussion.

\section{The Proposed Method}
\label{sec:methodlogy}

In this section, we motivate the idea of graph trend collaborative filtering  which models adaptive smoothness over the interactions, and introduce the corresponding algorithm.
Finally, we present the overall recommendation framework as well as the training process.

\newcommand{\Ein}{{\vE_{\text{in}}}}

\subsection{Design Motivation}
The non-adaptive and non-robust characteristics of embedding propagation in existing graph collaborative filtering methods have been discussed and verified in Section~\ref{sec:emb_prop}. To build a more reliable and robust recommender system, we aim to design a new collaborative filtering method that models the adaptive smoothness over the user-item interactions. 
Motivated by the concepts of trend filtering~\cite{kim2009ell_1, tibshirani2014adaptive} and graph trend filtering~\cite{wang2016trend}, we propose the following embedding smoothness objective for user-item graph in recommendations:
\begin{align}
\label{eq:gtf_obj}
    &\argmin_{\vE \in \RR^{(n+m)\times d}} \frac{1}{2} \|\vE-\Ein \|_F^2 + \lambda \| \tDelta\vE \|_1,
\end{align}
where $\Ein$ contains the initial embeddings of items and users. $\tDelta=\Delta \tD^{-\frac{1}{2}}$ is the incident matrix normalized by the square root of node degrees\footnote{The purpose of normalization is to ensure numerical stability and to reduce the dominating impacts of nodes with high degrees.}.
$\lambda$ is a hyper-parameter to control the strength of the embedding smoothness. The first term preserves the proximity with the initial embedding of users and items,
and the second term imposes embedding smoothness over the interactions $\cE$ since
\begin{align}
\| \tDelta\vE \|_1 = \sum_{(i,j) \in \cE} \|\frac{\ve_i}{\sqrt{d_i+1}}-\frac{\ve_j}{\sqrt{d_j+1}}\|_1.
\end{align}
It can be also rewritten as 
\begin{align}
\| \tDelta\vE \|_1 = \sum_{(i,j) \in \cE} \vW_{ij} \cdot \|\frac{\ve_i}{\sqrt{d_i+1}}-\frac{\ve_j}{\sqrt{d_j+1}}\|_2^2
\end{align}
where 
$$\vW_{ij} = \frac{\|\frac{\ve_i}{\sqrt{d_i+1}}-\frac{\ve_j}{\sqrt{d_j+1}}\|_1}{\|\frac{\ve_i}{\sqrt{d_i+1}}-\frac{\ve_j}{\sqrt{d_j+1}}\|_2^2}.$$
Comparing with the objective in Eq.~\eqref{eq:laplacian_smooth2}, $\| \tDelta \vE \|_1$ can be considered as a re-weighted version of $\tr(\vE^\top (\vI-\tA) \vE)$ where the smoothness over the interaction $(i, j)$ in the user-item graph is weighted by $\vW_{ij}$. Note that $\vW_{ij}$ is small if $\frac{\ve_i}{\sqrt{d_i+1}}$ is far away from $\frac{\ve_j}{\sqrt{d_j+1}}$, while $\vW_{ij}$ is large if
$\frac{\ve_i}{\sqrt{d_i+1}}$ is close to $\frac{\ve_j}{\sqrt{d_j+1}}$. Therefore, if an interaction happens between a user and an item that has significantly different embeddings, such interaction might be unreliable or might not reflect the actual preference of the user, so that it should be down weighted. 

The embedding smoothness prior defined in Eq.~\eqref{eq:gtf_obj} inherits the adaptive smoothness property from graph trend filtering that has been shown to adapt to inhomogeneity in the level of smoothness of signal and tend to provide piecewise polynomial (e.g., piecewise constant) signal over the graph~\cite{wang2016trend}. Specifically, the sparsity induced by the penalty $\|\tDelta \vE\|_1$ implies that many of the embedding differences $\frac{\ve_i}{\sqrt{d_i+1}}-\frac{\ve_j}{\sqrt{d_j+1}}$ are zeros across the edges $(i, j)$ if the embedding difference is sufficiently small. In other words, if the user and item embeddings are very similar, the filtering will make them even more closer. 
Moreover, the penalty is able to tolerate large embedding differences. 
That is to say, if the user and item embeddings are significantly different, the filtering will maintain the difference, which avoids improper smoothing. Therefore, it leads to more robust and adaptive propagation.

\subsection{Graph Trend Collaborative Filtering}

With the defined smoothness assumption in Eq.~\eqref{eq:gtf_obj}, we here derive an efficient algorithm to achieve the desired adaptive smoothness property for user-item graph in recommendations. The objective in Eq.~\eqref{eq:gtf_obj} is convex so that any local minimum is equivalent to the global minimum. Moreover, it is composed by a smooth function $f(\vE)=\frac{1}{2}\|\vE-\Ein\|_F^2$ and a nonsmooth function $g(\vE)=\lambda \|\tDelta \vE\|_1$, which are further coupled by the normalized incident matrix $\tDelta$ of the user-item graph. To decouple these two components, we first reformulate the problem as an equivalent constrained problem:
\begin{align}
    \min_{\vE, \vF} \frac{1}{2} \|\vE-\Ein\|_F^2 + \lambda \|\vF\|_1 ~~~
    \st ~~\vF = \tDelta \vE
    \nonumber
\end{align}
This is equivalent to the following saddle point problem by introducing the Lagrangian multipliers $\vY\in\RR^{|\cE|\times d}$:
\begin{align}
     \min_{\vE, \vF} \max_{\vY} \frac{1}{2} \|\vE-\Ein\|_F^2 + \langle \tDelta \vE, \vY \rangle + \|\vF\|_1 + \langle -\vF, \vY \rangle.
    \nonumber
 \end{align}
Using the conjugate function for $g(\vF) = \lambda \|\vF\|_1$:
\begin{align}
    g^*(\vY) \coloneqq \sup_{\vX} \langle \vY, \vX\rangle - g(\vX),
    \nonumber
\end{align}
 the saddle point problem can we rewritten as:
\begin{align}
\min_{\vE} \max_{\vY} \frac{1}{2} \|\vE-\Ein\|_F^2 + \langle \tDelta \vE, \vY \rangle - g^*(\vY).  
\nonumber
\end{align}

Then we can apply the primal-dual algorithm, Proximal Alternating Predictor-Corrector (PAPC)~\cite{loris2011generalization},
to obtain the following iterative solver:
\begin{align}
\bar \vE^{k+1} &= \vE^k - \gamma (\vE^k -\Ein) - \gamma \tDelta^\top \vY^k \label{eq:pd_step1} \\
\vY^{k+1} &= \prox_{\beta g^*} (\vY^k + \beta \tDelta \bar \vE^{k+1})  \label{eq:pd_step2} \\
\vE^{k+1} &= \vE^k - \gamma (\vE^k-\Ein) - \gamma \tDelta^\top \vY^{k+1} \label{eq:pd_step3} 
\end{align}
where $\prox_{\beta g^*}(\vX) = \argmin_{\vZ} \frac{1}{2} \|\vZ-\vX\|_F^2 + \beta g^*(\vZ)$, and $\gamma$ and $\beta$ are the primal and dual stepsizes to be specified later. The first step in Eq.~\eqref{eq:pd_step1} is a gradient descent step on the primal variable $\vE^k$, and the second step in Eq.~\eqref{eq:pd_step2} is a proximal dual ascent step on the dual variable $\vY^k$ based on the primal prediction $\bar \vE^{k+1}$. The third step in Eq.~\eqref{eq:pd_step3} is a gradient descent step on the primal variable based on the dual variable $\vY^{k+1}$. 

Next we describe the computation for the proximal operator $\prox_{\beta g^*}(\vX)$. By definition, the proximal operator of $\frac{1}{\beta} g(\vX)$ is defined as:
\begin{align}
\prox_{\frac{1}{\beta} g} (\vX) = \argmin_\vZ \frac{1}{2} \|\vZ-\vX\|_F^2 + \frac{\lambda}{\beta} \|\vZ\|_1,
\nonumber
\end{align}
which is equivalent to the soft-thresholding operator:
\begin{align}
    \big( S_{\frac{\lambda}{\beta}} (\vX) \big)_{ij} &= \sign(\vX_{ij}) \max(|\vX_{ij}| - \frac{\lambda}{\beta}, 0) \nonumber \\
    &= \vX_{ij} - \sign(\vX_{ij}) \min (|\vX_{ij}|, \frac{\lambda}{\beta}). \nonumber
\end{align}

Using the Moreau's decomposition principle~\cite{10.5555/2028633}
\begin{align}
    \vX = \prox_{\beta g^*} (\vX) + \beta \prox_{\frac{1}{\beta} g} (\frac{\vX}{\beta}),
    \nonumber
\end{align}
we can derive:
\begin{align}
    \prox_{\beta g^*} (\vX) &= \vX - \beta \prox_{\frac{1}{\beta} g} (\frac{\vX}{\beta}) \nonumber \\
    &= \vX - \beta \Big ( \frac{\vX}{\beta} - \sign(\frac{\vX}{\beta}) \min (|\frac{\vX}{\beta}|, \frac{\lambda}{\beta}) \Big ) \nonumber \\
    &= \sign(\vX) \min (|\vX|, \lambda). \nonumber
\end{align}

\noindent {\textbf{Parameter settings}.} We provide the following theorem to ensure the convergence of the derived iteration solver and provide guidance on the parameter settings.
\begin{theorem}
\label{thm:convergence}
Under the stepsize setting $\gamma = 1$ and $\beta =\frac{1}{2}$,
the iterations in Eq.~\eqref{eq:pd_step1}, ~\eqref{eq:pd_step2} and~\eqref{eq:pd_step3} are guaranteed to converge to the optimal solution of problem~\eqref{eq:gtf_obj}. 
\end{theorem}

\begin{proof}
 It is proved in the work~\cite{loris2011generalization} that the iterations in Eq.~\eqref{eq:pd_step1}, ~\eqref{eq:pd_step2} and~\eqref{eq:pd_step3} converge to the optimal solution if the parameters satisfy $\gamma < \frac{2}{L}$ and $\beta \leq \frac{1}{\gamma \| \tDelta \tDelta^\top \|_2}$, where $L=1$ is the Lipschitz constant of $\nabla \|\vE-\Ein\|_F^2=\vE-\Ein$. Note that $\|\tL\|_2=\|\tDelta^{\top}\tDelta\|_2=\|\tDelta\tDelta^{\top}\|_2 \leq 2 $~\cite{chung1997spectral}. Therefore, to ensure convergence, it is sufficient to choose $\gamma=1 < 2$ and $\beta=\frac{1}{2} \leq \frac{1}{\gamma \|\tDelta \tDelta^\top\|_2}$.
\end{proof}

Under the parameter setting implied by Theorem~\ref{thm:convergence}, the iterative solver can be simplified and summarized in Figure~\ref{fig:gtcf}. 
\begin{figure}[!ht]
 \centering
 \vskip -0.1in
\colorbox{gray!10}{
\centering 
 \begin{minipage}{0.95\columnwidth}
\begin{numcases}{}
\bar \vE^{k+1} = \Ein - \tDelta^\top \vY^k \label{eq:cf_step1} \nonumber \\
\bar \vY^{k+1} = \vY^k + \frac{1}{2} \tDelta \bar \vE^{k+1} \label{eq:cf_step2} \nonumber \\
\vY^{k+1} = \sign(\bar \vY^{k+1}) \min (|\bar \vY^{k+1}|, \lambda)  \label{eq:cf_step3} \nonumber
 \end{numcases}
  \begin{center}
  Output: $\vE^{k+1} = \Ein - \tDelta^\top \vY^{k+1}  \label{eq:cf_step4}  \nonumber$
 \end{center}
\end{minipage}
}
\vskip -0.12in
\caption{Graph Trend Collaborative Filtering (GTCF). 
 }
\vskip -0.12in
\label{fig:gtcf}
 \end{figure}

\noindent {\textbf{Complexity analysis}.} The proposed graph trend collaborative filtering in Figure~\ref{fig:gtcf} is efficient and scalable. The major computation cost comes from three sparse matrix multiplications, including $\tDelta^\top \vY^k, \tDelta \bar \vE^{k+1}$ and $\tDelta^\top \vY^{k+1}$. Note that the last one $\tDelta^\top \vY$ only needs to be computed in the output step but not in each iteration. The computation complexity for each iteration is in the order $(|\cE| d)$ where $|\cE|$ is the number of interactions between users and items and $d$ is the dimension of embedding. Therefore, the computation complexity is in the same order as the embedding propagation in Eq.~\eqref{eq:prop_matrix}.

\subsection{Graph Trend Filtering Networks}
In this subsection, we introduce the graph trend filtering networks (GTN) which includes four major components: 1) embedding initialization; 2) embedding filtering; 3) model prediction; and 4) model training.

\noindent {\textbf{Embedding initialization}.}
Given the general situations where only the user-item interaction data is available while there is no features for the user or item, we first randomly initialize the user and item embeddings:
$$\Ein = [\underbrace{\ve^0_1,\cdots,\ve^0_n}_{\text{users embeddings}}, \underbrace{\ve^0_{n+1},\cdots,\ve^0_{n+m}}_{\text{item embeddings}}]^\top \in \RR^{(n+m)\times d}.$$

In some applications where the features of users or items are available, we can either use the given features or concatenate them with randomly initialized embeddings. 

\noindent {\textbf{Embedding Filtering}.}
Given the initialized user and item embeddings $\Ein$, we apply the proposed graph trend collaborative filtering (GTCF as showed in Figure~\ref{fig:gtcf}) on $\Ein\in \RR^{(n+m)\times d}$ with $K$ iterations, such that we can extract the final representations of users and items: 
$$\vE^K = [\underbrace{\ve^K_1,\cdots,\ve^K_n}_{\text{users embeddings}}, \underbrace{\ve^K_{n+1},\cdots,\ve^K_{n+m}}_{\text{item embeddings}}]^\top \in \RR^{(n+m)\times d}$$
which encodes the user-item interaction information and follows the embedding smoothness assumption in Eq.~\eqref{eq:gtf_obj}.

\noindent {\textbf{Model Prediction}.} 
Finally, the model prediction on whether a user $u$ is interested in an item $i$ is defined as the inner product between their final representations $\vE^K$: 
\begin{equation}
    \hat{y}_{ui} = {\ve^K_u}^T \ve^K_i. \nonumber
\end{equation}
 The prediction will be used as the ranking score for recommendation generation.

\noindent  \textit{\textbf{Model Training}.} 
The model parameters to be learned are the initial user and item embeddings $\Ein$. To train a good set of model parameters of our proposed method, we need to identify an objective function to optimize.  In this paper, we adopt the pairwise BPR loss~\cite{Rendle2009BPRBP} to optimize model parameters based on the implicit feedback in recommendations, assuming that score predictions on the observed interactions should be higher than the unobserved interactions: 
\begin{align}
 \cL_{\text{BPR}} = \sum_{(u,i,j) \in O} -\ln \sigma(\hat{y}_{ui}-\hat{y}_{uj}) + \alpha \|\Ein\|_F^2, \nonumber
\end{align}
where $\alpha$ controls the $L_2$ regularization to reduce the risk of overfitting. $\sigma(\cdot)$ is the Sigmoid function. The set $ O=\{(u,i,j)|(u,i)\in \cE, (u,j) \notin \cE \}$ denotes the pairwise training data, where $(u,i) \in \cE$ indicates that the interactions between user $v_u$ and item $v_j$ is observed (positive), while $(u,j) \notin \cE$ indicates that the interaction between user $v_u$ and item $v_j$ is not observed (negative).

\section{Experiment}
\label{sec:Experiments}
 
 

 \subsection{\textbf{Experimental Settings}}

\subsubsection{\textbf{Datasets}}
To demonstrate the effectiveness of the proposed method, we conduct comprehensive experiments on four widely used benchmark datasets in recommender systems, including Gowalla, Yelp2018, Amazon-Book, and LastFM. The first three datasets are released by NGCF~\cite{Wang2019NeuralGC} and LightGCN~\cite{He2020LightGCNSA}; the LastFM dataset is provided by  KGAT~\cite{wang2019kgat}. 
For fair comparison, we closely follow the experimental setting of LightGCN\footnote{https://github.com/gusye1234/LightGCN-PyTorch\label{fn:lightgcn}}. In the training phase, each observed user-item interaction is treated as a positive instance, while we conduct negative sampling  to randomly sample one unobserved item and pair it with the user as a negative instance.
The statistical summary of these four datasets can be found in Table~\ref{tab:dataset}.  

\begin{table}[htbp]
\centering
\caption{Basic statistics of benchmark datasets.}
\label{tab:dataset}
\vskip -0.1050in
\scalebox{0.7}
{
\begin{tabular}{|c||c|c|c|}
\hline
\multirow{2}{*}{\textbf{Datasets}} & \multicolumn{3}{c|}{\textbf{User-Item Interaction}}           \\ \cline{2-4} 
                                   & \textbf{\#Users} & \textbf{\#Items} & \textbf{\#Interactions} \\ \hline 
\textbf{Gowalla}                   & 29,858          & 40,981          & 1,027,370               \\ \hline
\textbf{Yelp2018}                  & 31,668          & 38,048          & 1,561,406               \\ \hline
\textbf{Amazon-Book}               & 52,643          & 91,599          & 2,984,108               \\ \hline
\textbf{LastFM}                    & 23,566           & 48,123           & 3,034,796               \\ \hline
\end{tabular}
}
\vskip -0.120in
\end{table}

\begin{table*}[htbp]
\centering
\vskip -0.20in
\caption{The comparison of overall performance.}
\vskip -0.1050in
\label{tab:comparsion_all}
\scalebox{0.80}
{
\begin{tabular}{|c|c|c|c|c|c|c|c|c|c|}
\hline
\multicolumn{2}{|c|}{\textbf{Datasets}}              & \multicolumn{2}{c|}{\textbf{Gowalla}} & \multicolumn{2}{c|}{\textbf{Yelp2018}} & \multicolumn{2}{c|}{\textbf{Amazon-Book}} & \multicolumn{2}{c|}{\textbf{LastFM}}  \\ \hline
\multicolumn{2}{|c|}{\textbf{Metrics}}               & \textbf{Recall@20} & \textbf{NDCG@20} & \textbf{Recall@20}  & \textbf{NDCG@20} & \textbf{Recall@20}   & \textbf{NDCG@20}   & \textbf{Recall@20} & \textbf{NDCG@20} \\ \hline \hline
\multirow{8}{*}{\textbf{Method}} & \textbf{MF}       & 0.1299             & 0.111            & 0.0436              & 0.0353           & 0.0252               & 0.0198             & 0.0725             & 0.0614           \\ \cline{2-10} 
                                 & \textbf{NeuCF}    & 0.1406             & 0.1211           & 0.045               & 0.0364           & 0.0259               & 0.0202             & 0.0723             & 0.0637           \\ \cline{2-10} 
                                 & \textbf{GC-MC}    & 0.1395             & 0.1204           & 0.0462              & 0.0379           & 0.0288               & 0.0224             & 0.0804             & 0.0736           \\ \cline{2-10} 
                                 & \textbf{NGCF}     & 0.156              & 0.1324           & 0.0581              & 0.0475           & 0.0338               & 0.0266             & 0.0774             & 0.0693           \\ \cline{2-10} 
                                 & \textbf{Mult-VAE} & 0.1641             & 0.1335           & 0.0584              & 0.045            & 0.0407               & 0.0315             & 0.078              & 0.07             \\ \cline{2-10} 
                                 & \textbf{DGCF}     & 0.1794             & 0.1521           & 0.064               & 0.0522           & 0.0399               & 0.0308             & 0.0794             & 0.0748           \\ \cline{2-10} 
                                 & \textbf{LightGCN} & 0.1823             & 0.1553           & 0.0649              & 0.0525           & 0.042                & 0.0327             & 0.085              & 0.076            \\ \cline{2-10} 
                                 & \textbf{\ourname{} (Ours)}      & 0.187              & 0.1588           & 0.0679              & 0.0554           & 0.045                & 0.0346             & 0.0932             & 0.0857           \\ \hline
\multicolumn{2}{|c|}{\textbf{Relative Improvement (\%)}}                & 2.59               & 2.26             & 4.62                & 5.59             & 7.15                 & 5.95               & 9.61               & 12.77            \\ \hline
\end{tabular}
}
\vskip -0.150in
\end{table*}

\subsubsection{\textbf{Baselines}}

To evaluate the effectiveness, we compare our proposed method with the following baselines:

\begin{itemize}[leftmargin=*] 
\item MF~\cite{Rendle2009BPRBP}: This is the most classic CF method optimized by the Bayesian Personalized Ranking (BPR) loss.

\item NeuCF~\cite{He2017NeuralCF}: This method is the very first DNNs based collaborative filtering which stacks multiple hidden layers to capture their non-linear user and item interactions.

\item GC-MC~\cite{berg2017graph}:  The method proposes  a graph auto-encoder framework for recommendations and  generates representations of  users  and  items  through  a  form  of  differentiable  message   passing   on   the   user-item  graph. 

\item Mult-VAE~\cite{Liang2018VariationalAF}\footnote{https://github.com/dawenl/vae\_cf}: This method  develops a variant of Variational AutoEncoder (VAE) for item-based  CF. 

\item NGCF~\cite{Wang2019NeuralGC}:  This method is the very first GNNs based CF to incorporate the high-order connectivity of user-item interactions for recommendations.

\item DGCF~\cite{Wang2020DisentangledGC}\footnote{https://github.com/xiangwang1223/disentangled\_graph\_collaborative\_filtering}: This method is a disentangled CF to learn representations for different latent user intentions via GNNs, so as to  exploit diverse relationships among users and items.

\item LightGCN~\cite{He2020LightGCNSA}\textsuperscript{\ref{fn:lightgcn}}:  The method is a state-of-the-art CF based on GNNs, which is an extension of NGCF~\cite{Wang2019NeuralGC} by removing feature transformation and nonlinear activation, achieving the state-of-the-art performance for recommendations. 
\end{itemize}

Note that in this work, we don't compare with SGL~\cite{wu2021self}, one of the state-of-art recommendation methods with self-supervised learning (SSL) enhanced on loss objective. 
What's more, the contribution of SGL is totally orthogonal to the contributions in our work, since we are designing a new collaborative filter for recommendations, regardless of the loss function. We will explore combination between our proposed GTN and SSL in the future.


\subsubsection{\textbf{Evaluation Metrics}}

In order to evaluate the quality of recommendation results for ranking task, we adopt two widely used evaluation metrics: Recall@$K$ and Normalized Discounted Cumulative Gain (NDCG@$K$)~\cite{Wang2019NeuralGC,He2020LightGCNSA}. 
By default, we set the value of $K$ as 20. Note that higher values of Recall@$K$ and NDCG@$K$ indicate better performance for recommendations. We report the average metrics for all users in the test set.

\subsubsection{\textbf{Parameter Settings}}

We implement our proposed model based on PyTorch\footnote{\url{https://github.com/wenqifan03/GTN-SIGIR2022}}.
For  embedding size $d$, we tested the  value of $\left \{16, 32, 64, 128, 256, 512  \right \}$.  The batch size and learning rate were searched in $\left \{ 128, 512, 1024, 2048  \right \}$ and $\left \{0.0005, 0.001, 0.005, 0.01, 0.05, 0.1 \right \}$, respectively.   We employ the Adam~\cite{Adam}  to optimize the objective function in a mini-batch manner, where we  first randomly sample mini-batch observed entries and then we pair it with an unobserved item $j$ that is uniformly sampled from $\cN(v_u)$  for each observed  interaction $(u,i)$. 
Note that we adopt the default hyper-parameters as suggested by the corresponding papers for all baselines, and  we closely follow the settings of the NGCF and LightGCN\textsuperscript{\ref{fn:lightgcn}} for a fair comparison.

\subsection{\textbf{Performance Comparison of Recommender Systems}}

We first compare the performance of all recommendation methods. Table~\ref{tab:comparsion_all} reports the overall performance comparison with regard to Recall@20 and NDCG@20 metrics on four datasets. We make the following observations: 
\begin{itemize}[leftmargin=*] 
\item   \vspace{-0.05in}
Our proposed \ourname{} achieves the best performance, and consistently outperforms all the baselines across all datasets in terms of all metrics. 
For instance, it significantly improves the strongest baselines by 12.77\% on NDCG@20 and  9.61\% on Recall@20 in LastFM dataset. 
Such improvement demonstrates the effectiveness of our proposed method and we attribute it to the imposed local adaptive smoothness over the user-item interactions. The local adaptive smoothness greatly improves  the adaptiveness of the aggregation scheme in GNN-based models and alleviates the vulnerability to noisy user-item interactions. Hence, it justifies the great potential of studying graph trend filtering in recommender systems.

\item As an earlier CF method, Matrix Factorization (MF) shows relatively poor performance for recommendations, while others like NeuCF and GNN-based methods show much stronger performance. This suggests the powerful capability on representation learning of deep neural networks techniques in recommendations. 

\item The results of NGCF and DGCF are at the same level, being better than MF and NeuCF methods. It demonstrates the importance of modeling the local structure of a node in the user-item graph to enhance the representation learning of users and items. By stacking propagation layers, most GNN-based recommendation methods can effectively capture collaborative signals via high-order connectivity.

\item  
LightGCN is a strong baseline as it almost beats MF, NeuCF, GC-MC, NGCF, Mult-VAE, and DGCF in all datasets. It inherits the characteristics of NGCF as well as the embedding propagation scheme of GNNs.
However, during the embedding propagation process, LightGCN tends to enforce the user embedding and item embedding to be close if there exist interactions between them. It can result in inappropriate smoothness caused by unreliable interactions such as unintended clicks, thus leading to sub-optimal performance. 
On the contrary, the proposed \ourname{} provides adaptive embedding smoothness over the user-item interactions by taking advantage of graph trend filtering. As a result, our proposed method  routinely achieves better recommendation performance than LightGCN in terms of Recall@20 and NDCG@20 metrics on four datasets.  

\end{itemize}

\begin{table*}[htbp]
\centering
\vskip -0.150in
\caption{The effect of the number of layers.}
\vskip -0.150in

\label{tab:num_layers}
\scalebox{0.8}
{
\begin{tabular}{|c|c|c|c|c|c|c|c|c|c|}
\hline
\multicolumn{2}{|c|}{\textbf{Datasets}}                        & \multicolumn{2}{c|}{\textbf{Gowalla}} & \multicolumn{2}{c|}{\textbf{Yelp2018}} & \multicolumn{2}{c|}{\textbf{Amazon-Book}} & \multicolumn{2}{c|}{\textbf{LastFM}}  \\ \hline
\multicolumn{2}{|c|}{\textbf{Method}}                          & \textbf{Recall@20} & \textbf{NDCG@20} & \textbf{Recall@20}  & \textbf{NDCG@20} & \textbf{Recall@20}   & \textbf{NDCG@20}   & \textbf{Recall@20} & \textbf{NDCG@20} \\ \hline \hline
\multirow{3}{*}{\textbf{1 Layer}} & \textbf{NGCF}     & 0.1548             & 0.1306           & 0.0544              & 0.0439           & 0.031                & 0.024              & 0.0751             & 0.0671           \\ \cline{2-10} 
                                  & \textbf{LightGCN} & 0.1753             & 0.149            & \textbf{0.0628}              & \textbf{0.0517}           & 0.0381               & 0.0298             & 0.0865             & 0.0777           \\ \cline{2-10} 
                                  & \textbf{\ourname{} (Ours)}      & \textbf{0.1789}             & \textbf{0.1499 }          & 0.06              & 0.0493           &\textbf{ 0.0443}               & \textbf{0.0341}             & \textbf{0.0888}             & \textbf{0.0806}           \\ \hline\hline
\multirow{3}{*}{\textbf{2 Layer}} & \textbf{NGCF}     & 0.1525             & 0.1313           & 0.0569              & 0.0463           & 0.0328               & 0.0252             & 0.0766             & 0.0681           \\ \cline{2-10} 
                                  & \textbf{LightGCN} & 0.1752             & 0.1524           & 0.0622              & 0.0505           & 0.0413               & 0.0312             & 0.0874             & 0.0788           \\ \cline{2-10} 
                                  & \textbf{\ourname{} (Ours)}      & \textbf{0.1863}             & \textbf{0.1576 }          & \textbf{0.0647}              & \textbf{0.0528}           & \textbf{0.0455}               & \textbf{0.0351 }            & \textbf{0.0913}             &\textbf{ 0.0831 }          \\ \hline\hline
\multirow{3}{*}{\textbf{3 Layer}} & \textbf{NGCF}     & 0.156              & 0.1324           & 0.0581              & 0.0475           & 0.0338               & 0.0266             & 0.0774             & 0.0693           \\ \cline{2-10} 
                                  & \textbf{LightGCN} & 0.1823             & 0.1553           & 0.0649              & 0.0525           & 0.042                & 0.0327             & 0.085              & 0.076            \\ \cline{2-10} 
                                  & \textbf{\ourname{} (Ours)}      & \textbf{0.187}              & \textbf{0.1588 }          & \textbf{0.0679}              & \textbf{0.0554}           & \textbf{0.045}                & \textbf{0.0346 }            & \textbf{0.0932 }            &\textbf{ 0.0857 }          \\ \hline\hline
\multirow{3}{*}{\textbf{4 Layer}} & \textbf{NGCF}     & 0.1564             & 0.1321           & 0.0568              & 0.0461           & 0.0344               & 0.0261             & 0.0726             & 0.0652           \\ \cline{2-10} 
                                  & \textbf{LightGCN} & 0.1818             & \textbf{0.1532}           & 0.065               & 0.0523           & 0.0404               & 0.0311             & 0.0822             & 0.0736           \\ \cline{2-10} 
                                  & \textbf{\ourname{} (Ours)}      & \textbf{0.1842}             & {0.1531}           &\textbf{ 0.0663 }             & \textbf{0.0535}           & \textbf{0.0418}               & \textbf{0.0322 }            & \textbf{0.0928}             & \textbf{0.085}            \\ \hline\hline
\multirow{3}{*}{\textbf{5 Layer}} & \textbf{NGCF}     & 0.1465             & 0.122            & 0.0557              & 0.0452           & 0.0322               & 0.0247             & 0.067              & 0.0607           \\ \cline{2-10} 
                                  & \textbf{LightGCN} & 0.1742             & 0.1464           & 0.0627              & 0.0509           & 0.0388               & 0.0299             & 0.0768             & 0.0685           \\ \cline{2-10} 
                                  & \textbf{\ourname{} (Ours)}      & \textbf{0.1789 }            & \textbf{0.1475}           &\textbf{ 0.0637}              & \textbf{0.0511}           & \textbf{0.0391}               & \textbf{0.03 }              & \textbf{0.0921}             & \textbf{0.0843}           \\ \hline
\end{tabular}
}
\vskip -0.1050in
\end{table*}

\subsection{\textbf{Study of Robustness and  Adaptive Smoothness }}

\subsubsection{\textbf{Robustness against random perturbation}}
As the user-item graph may contain a large portion of unreliable or noisy edges (e.g., random or bait clicks) which cannot reflect the actual  satisfaction and preferences of users in real-world systems, developing a robust GNN-based recommender system against such noise is essential in providing high-quality recommendations and improving user online experience.

We now investigate the robustness of different GNN-based recommendation methods under different severity of noise, i.e., different ratios of noisy edges in the user-item graph. Specifically, we perturb the test data by randomly injecting a certain ratio of edges and evaluate the recommendation performance of the models trained on clean data. In addition, we vary the ratio of noisy edges in the range of \{5\%, 6\%, 8\%, 10\%, 15\%, 20\%, 30\%, 50\%\}. Note that we also call the aforementioned noise ratio as perturbation rate.

The main results are shown in Figure~\ref{fig:perturbation_rates}. From the figure, we can see that our proposed method \ourname{} consistently outperforms all other baselines (i.e., NGCF and LightGCN), implying that our method can effectively resist random noises in user-item graph for recommendations. This promising performance can be attributed to the consideration of adaptive embedding smoothness over user-item interactions that are modeled by the proposed graph trend collaborative filtering. It is worth noting that 
under higher noise ratio, \ourname{} method even outperforms the baselines by a larger margin. Specifically, the performance of LightGCN significantly drops when the perturbation rate is larger than 20\% especially on Gowalla and Yelp2018; the performance of NGCF is almost the worst though relatively stable against random perturbation. On the contrary, \ourname{} maintains its high performance across different noise ratios, which demonstrates its superior robustness.

\begin{figure*}[htbp]
\vskip -0.1050in

\centering
{\subfigure[Gowalla - Recall@20]
{\includegraphics[width=0.22\linewidth]{{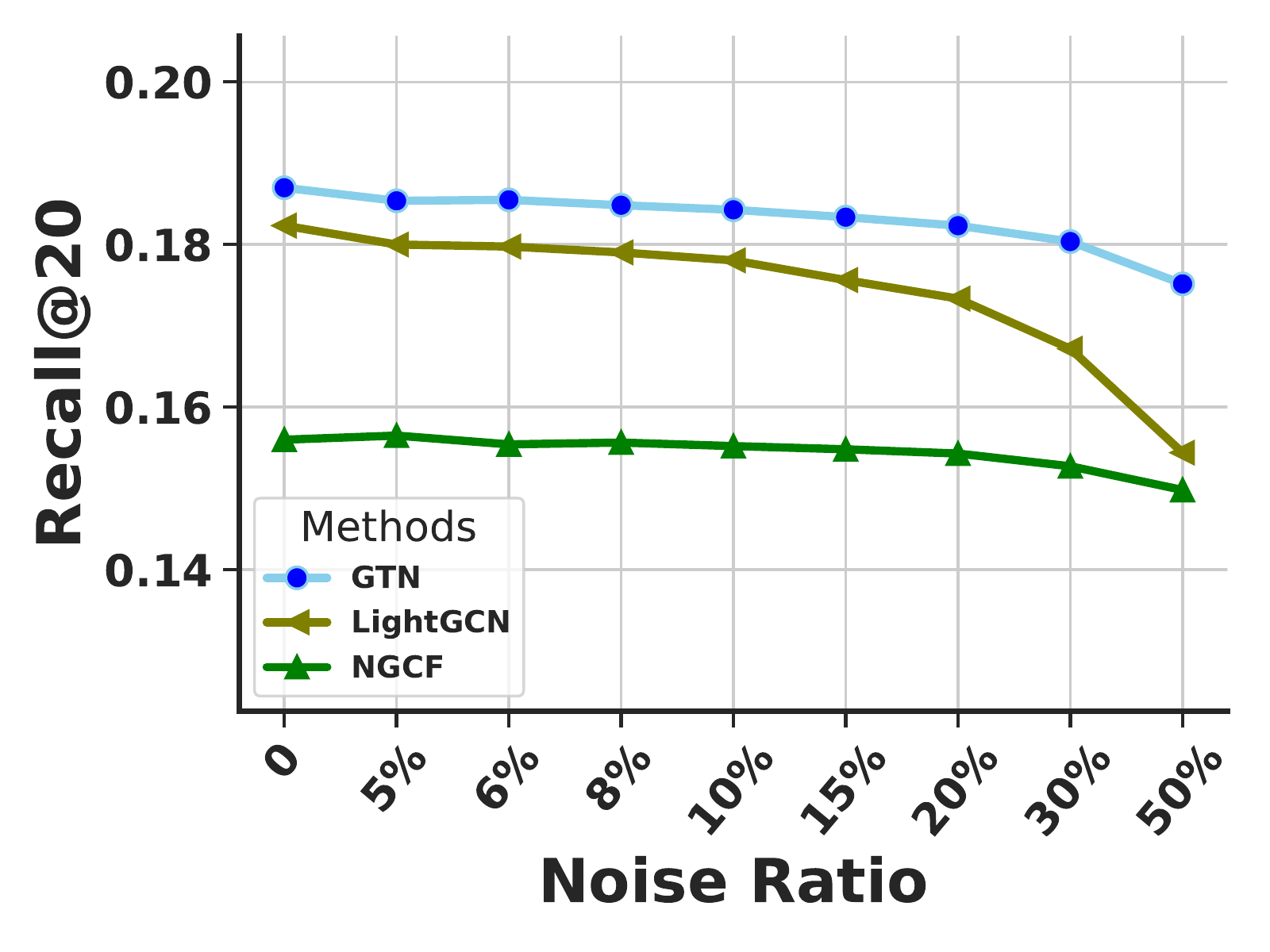}}}}
{\subfigure[Yelp2018 - Recall@20]
{\includegraphics[width=0.22\linewidth]{{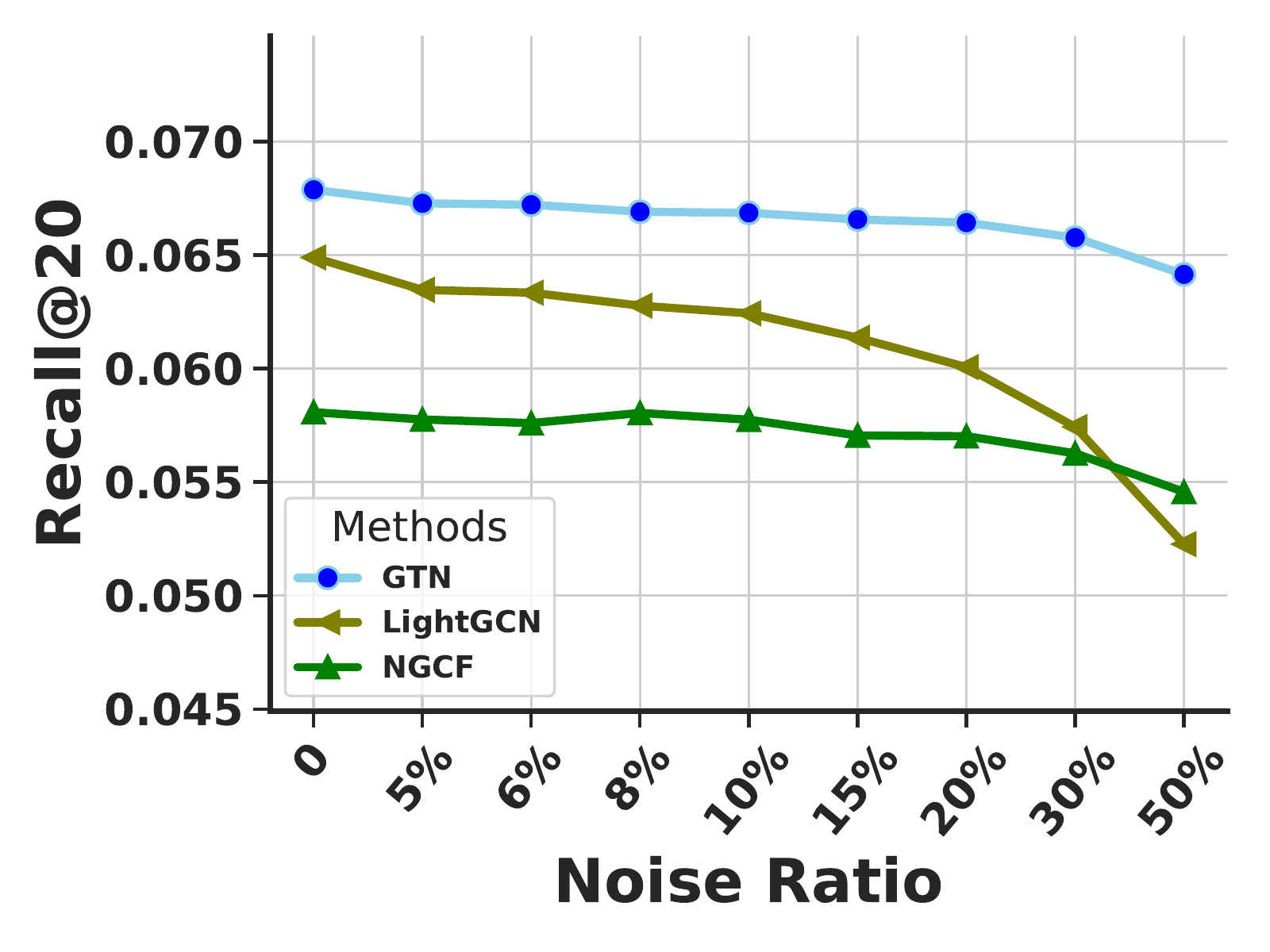}}}}
{\subfigure[Amazon-book - Recall@20]
{\includegraphics[width=0.22\linewidth]{{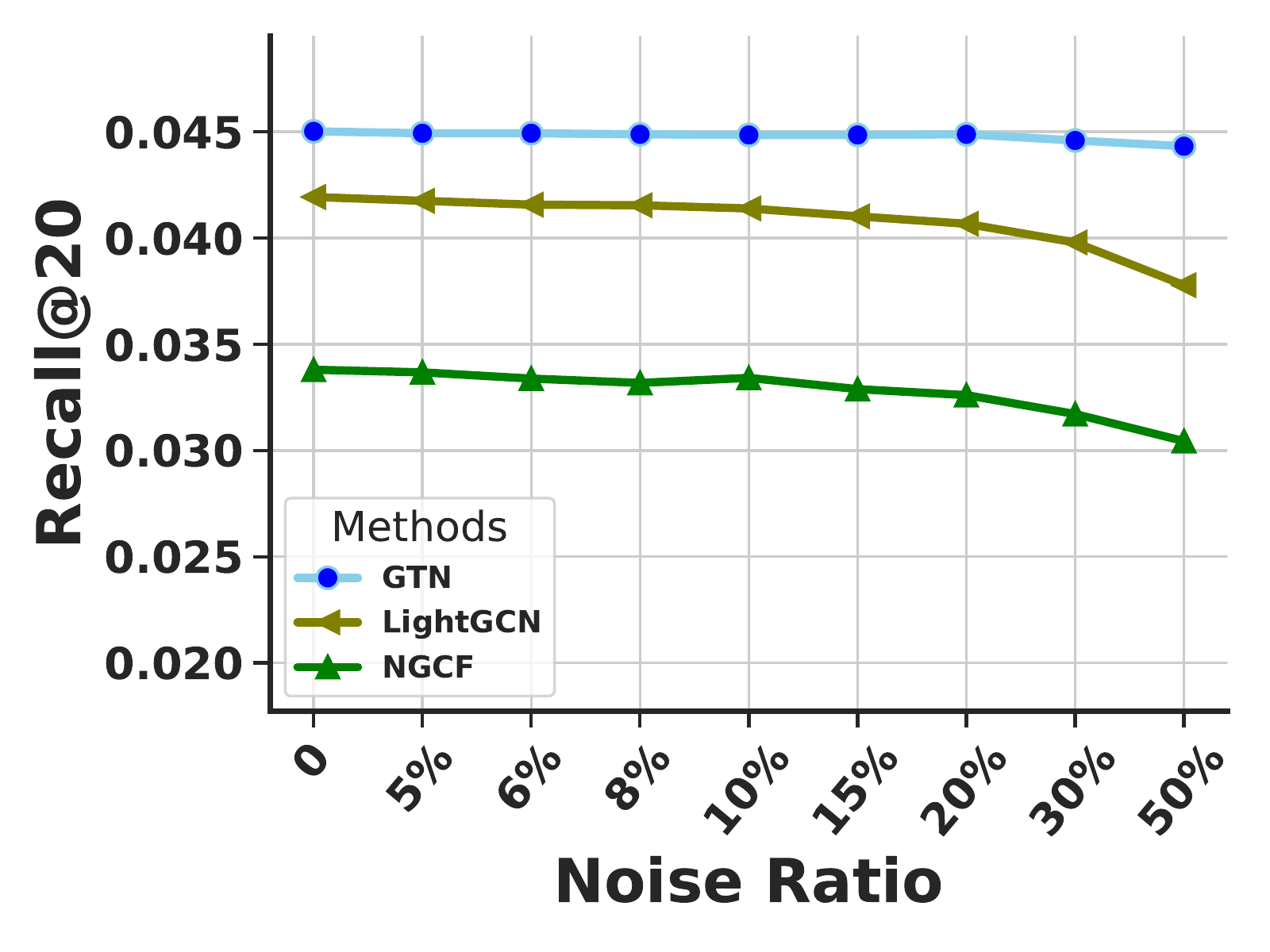}}}}
{\subfigure[LastFM - Recall@20]
{\includegraphics[width=0.22\linewidth]{{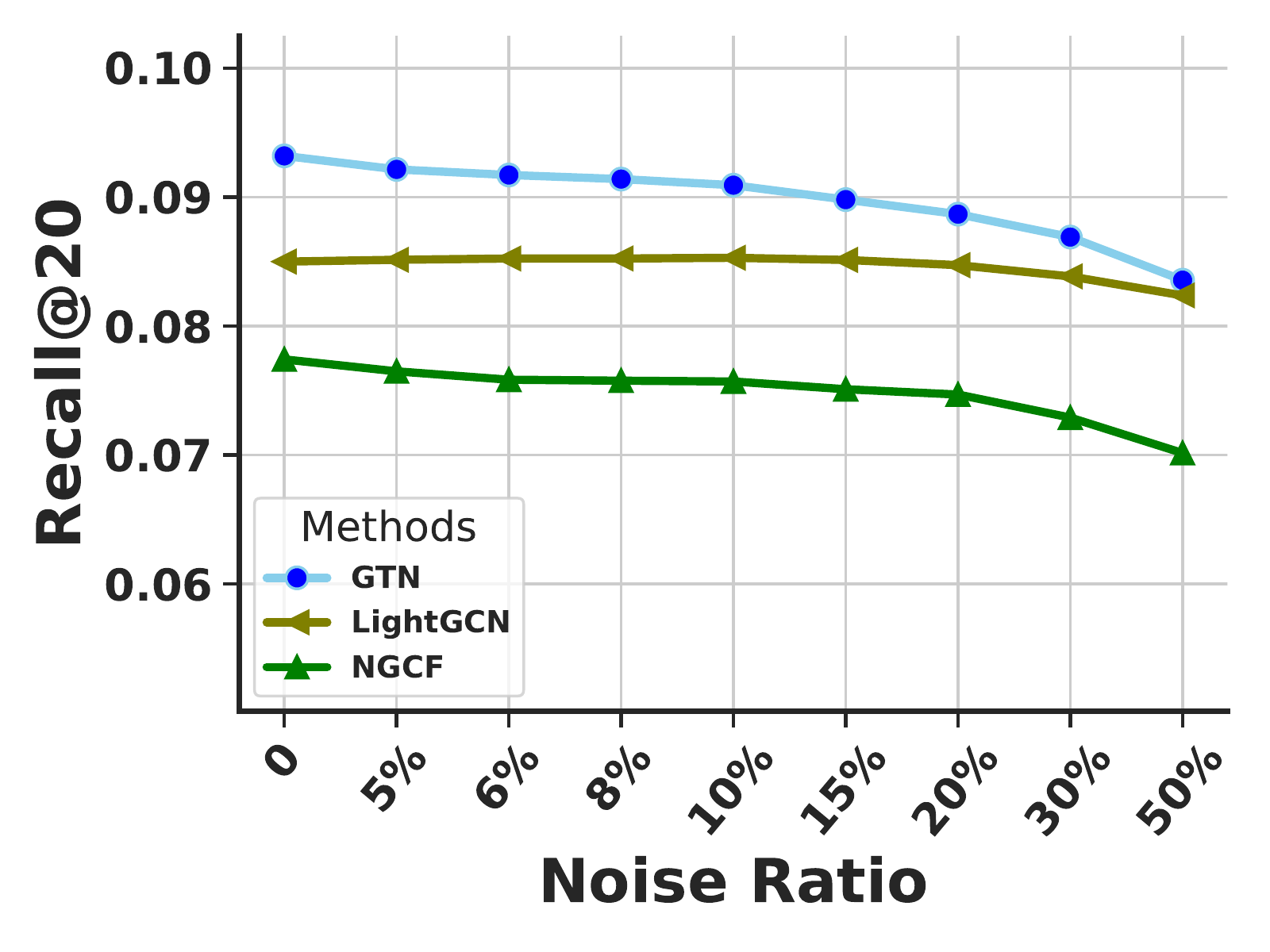}}}}

{\subfigure[Gowalla - NDCG@20]
{\includegraphics[width=0.22\linewidth]{{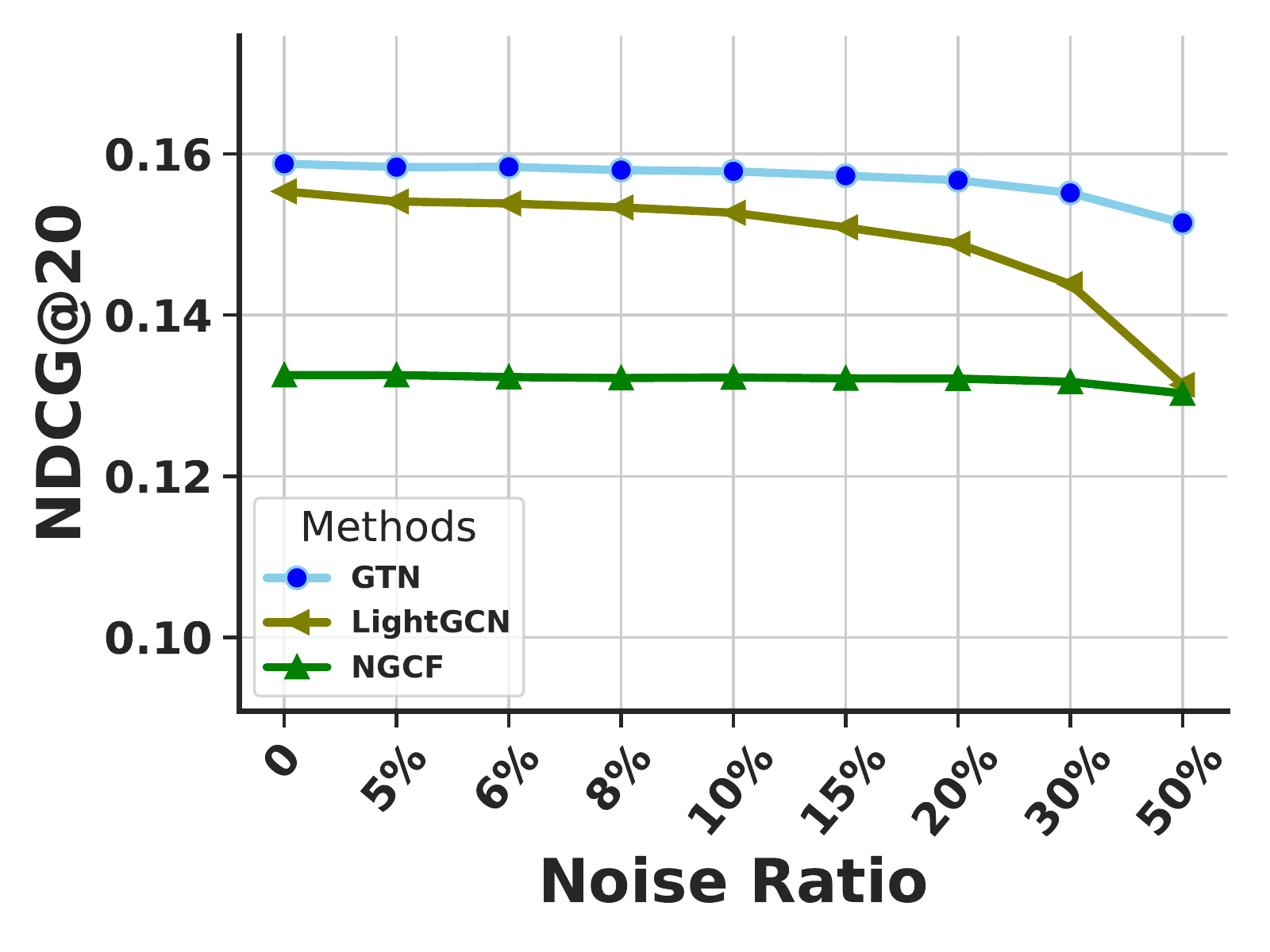}}}}
{\subfigure[Yelp2018 - NDCG@20]
{\includegraphics[width=0.22\linewidth]{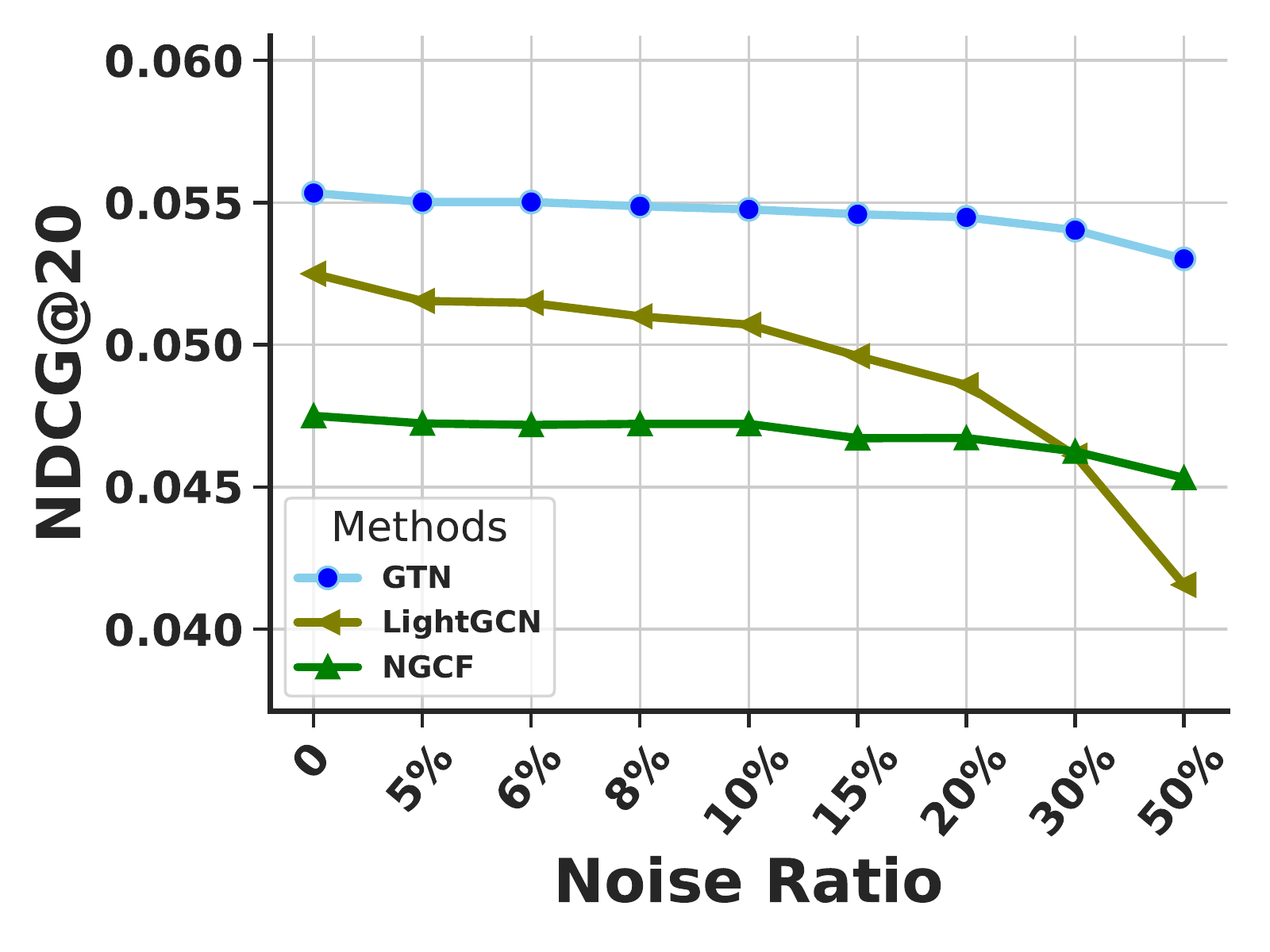}}}
{\subfigure[Amazon-book - NDCG@20]
{\includegraphics[width=0.22\linewidth]{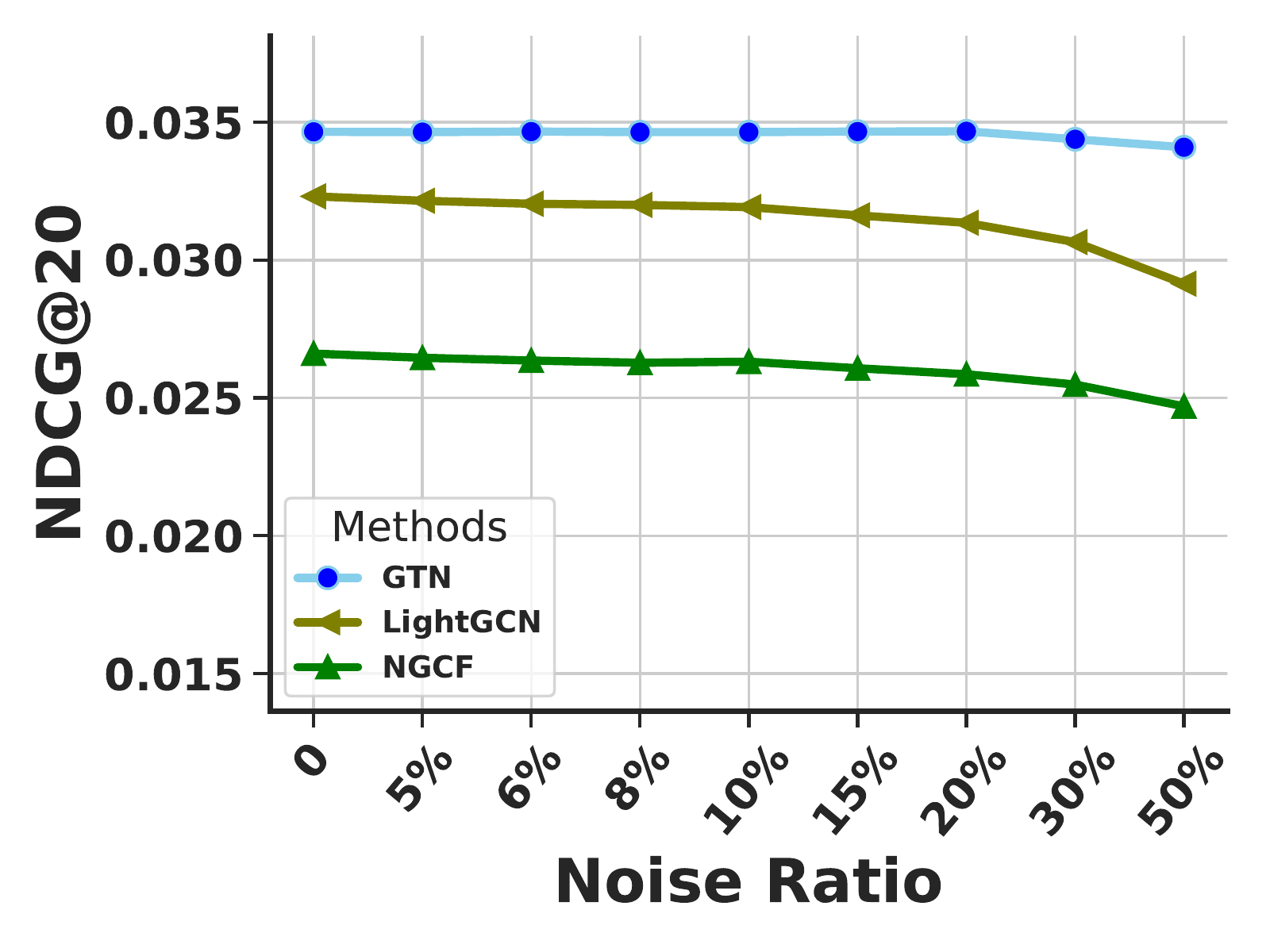}}}
{\subfigure[LastFM - NDCG@20]
{\includegraphics[width=0.22\linewidth]{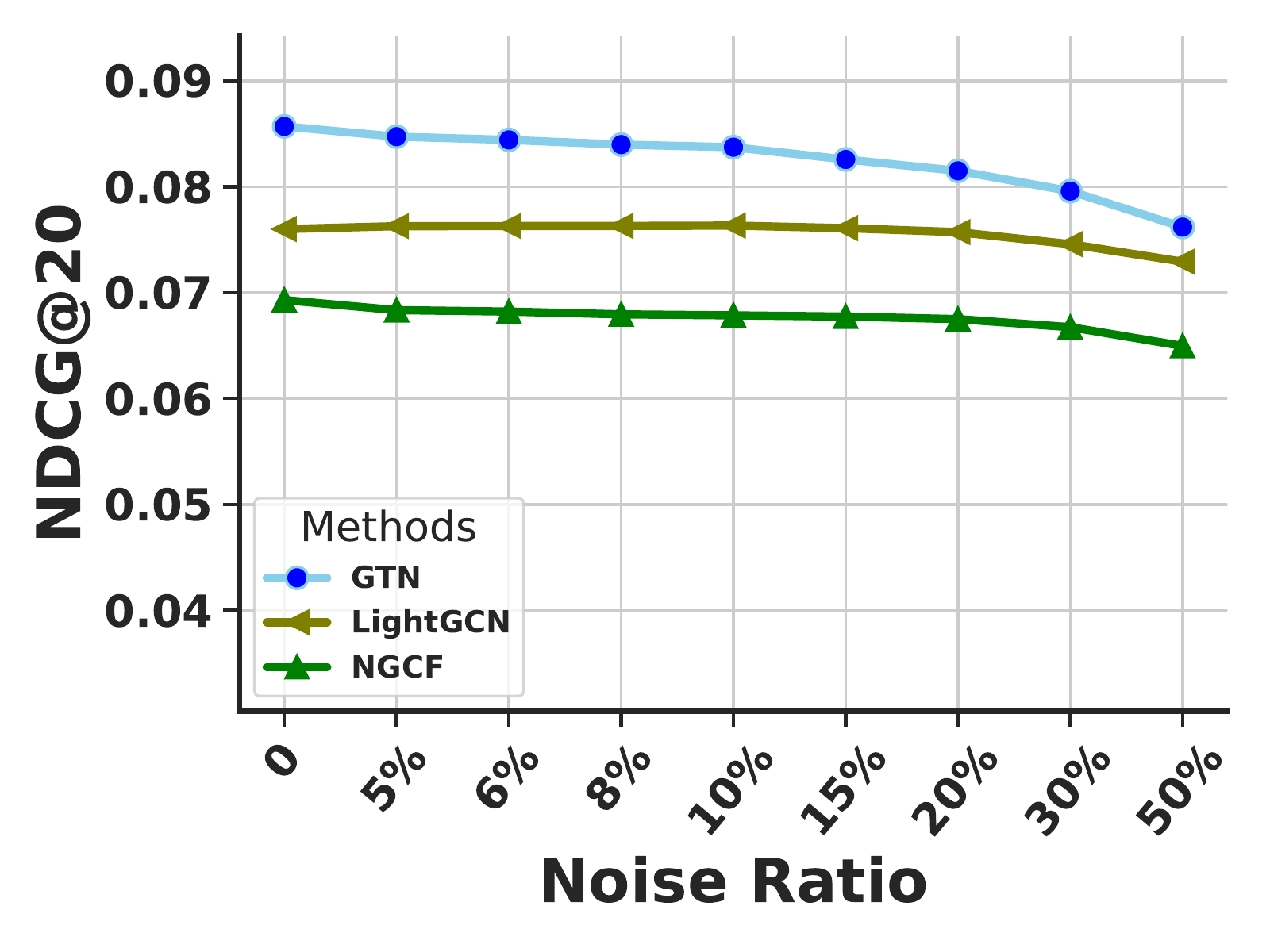}}}
\vskip -0.2in
\caption{Recommendation performance under different perturbation rates.}\label{fig:perturbation_rates}
\vskip -0.1in
\end{figure*}

\begin{figure*}[htbp]
\vskip -0.1in
\centering
{\subfigure[Gowalla - Recall@20]
{\includegraphics[width=0.22\linewidth]{{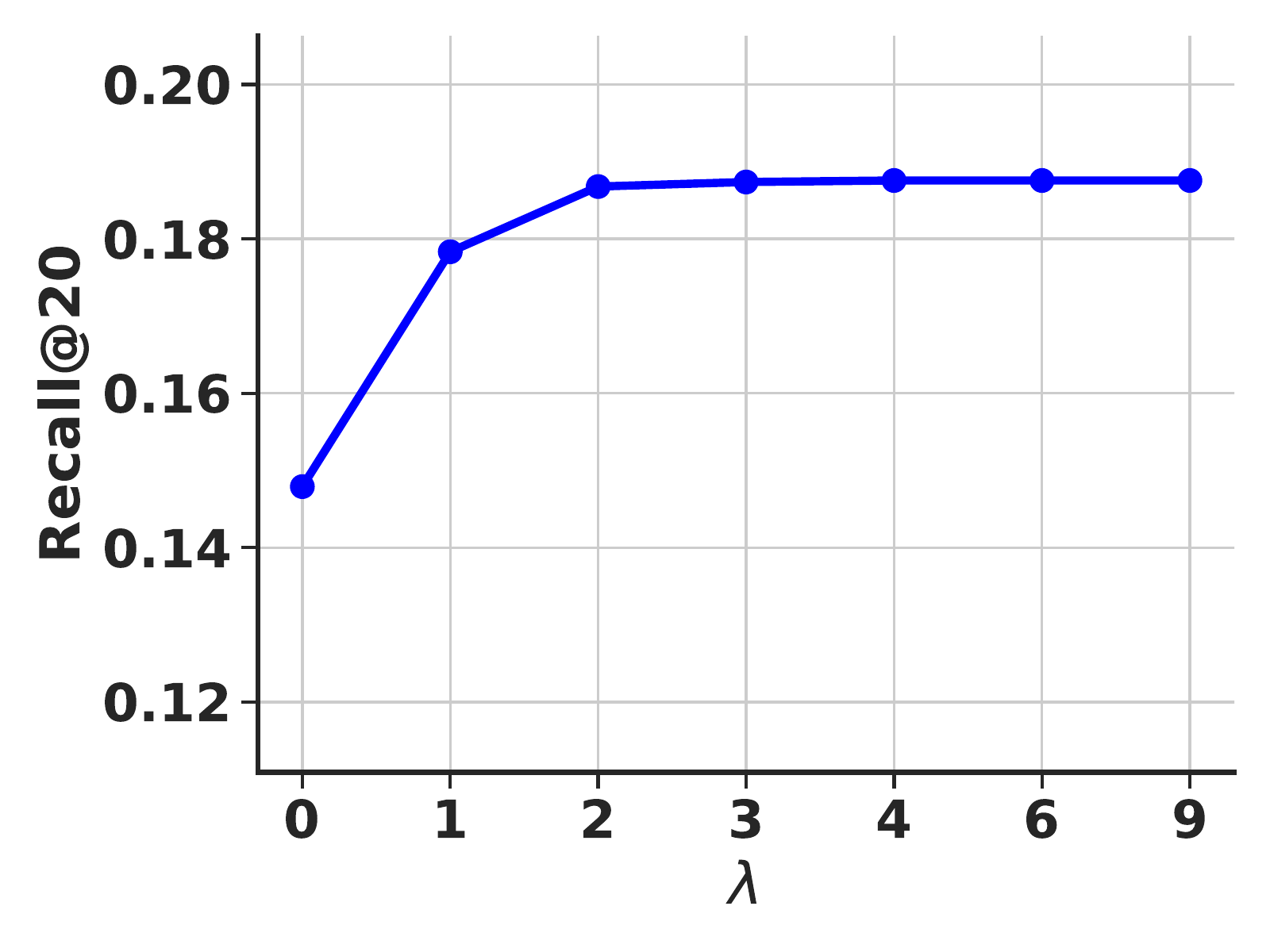}}}}
{\subfigure[Yelp2018 - Recall@20]
{\includegraphics[width=0.22\linewidth]{{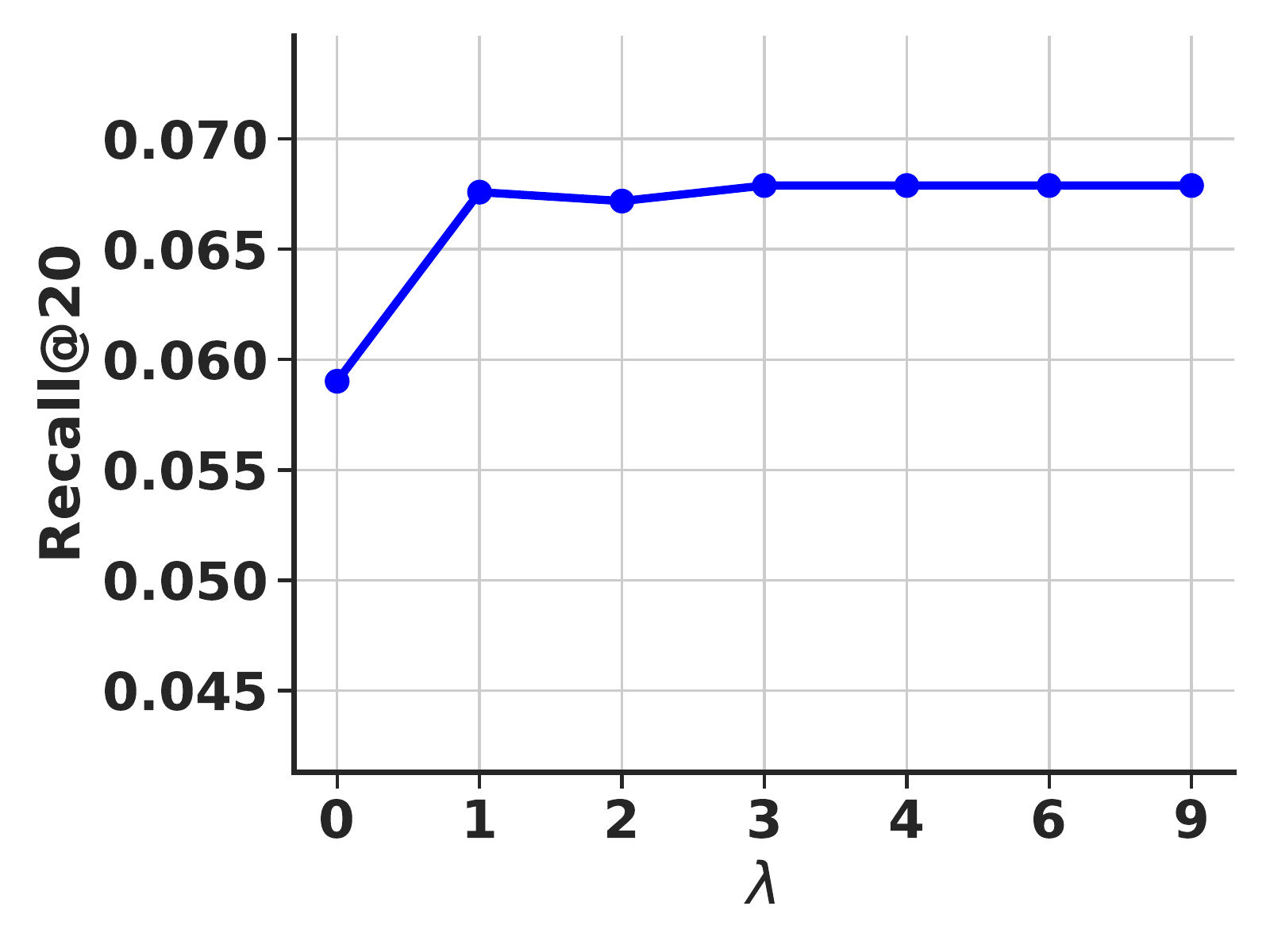}}}}
{\subfigure[Amazon-book - Recall@20]
{\includegraphics[width=0.22\linewidth]{{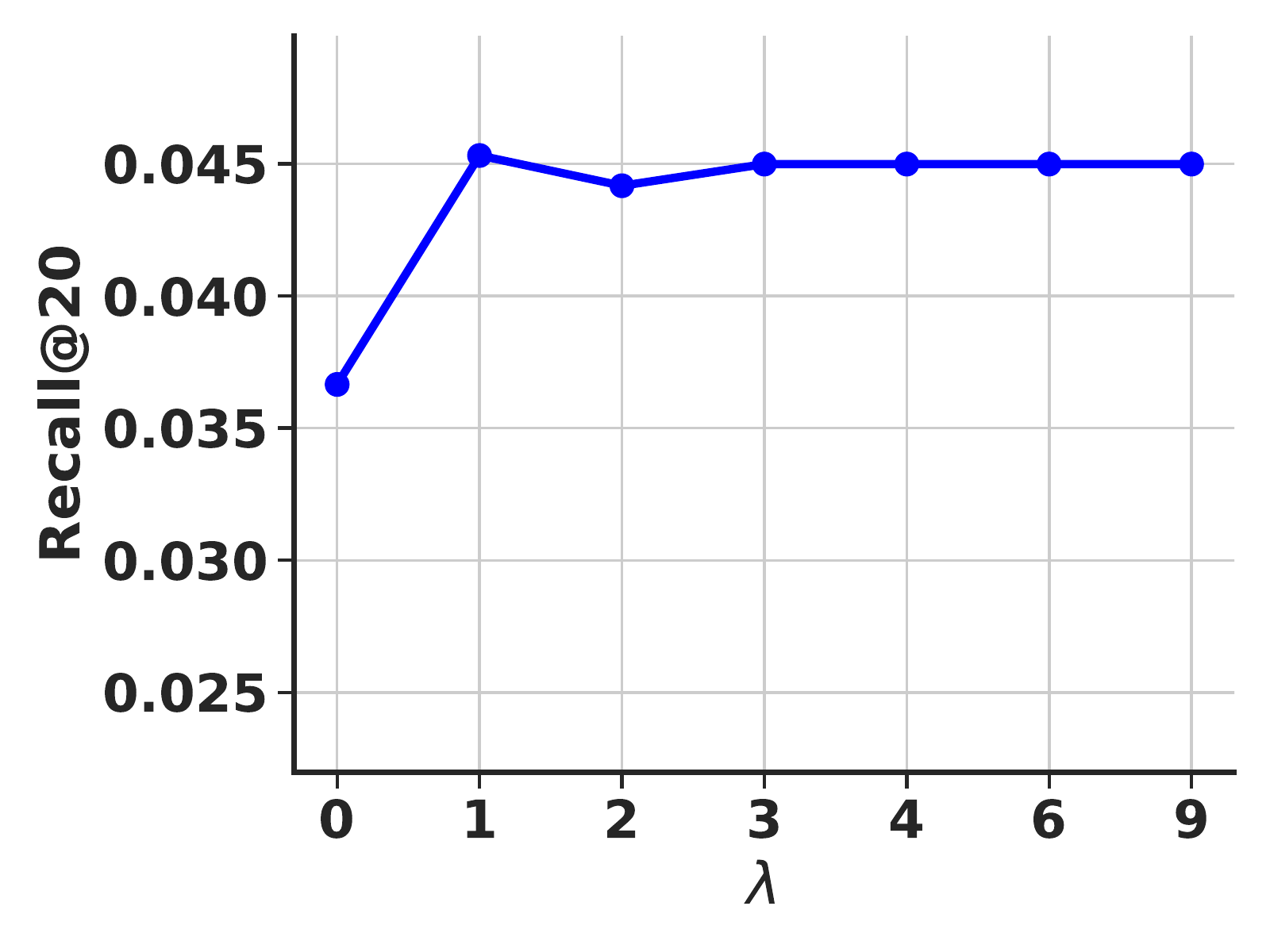}}}}
{\subfigure[LastFM - Recall@20]
{\includegraphics[width=0.22\linewidth]{{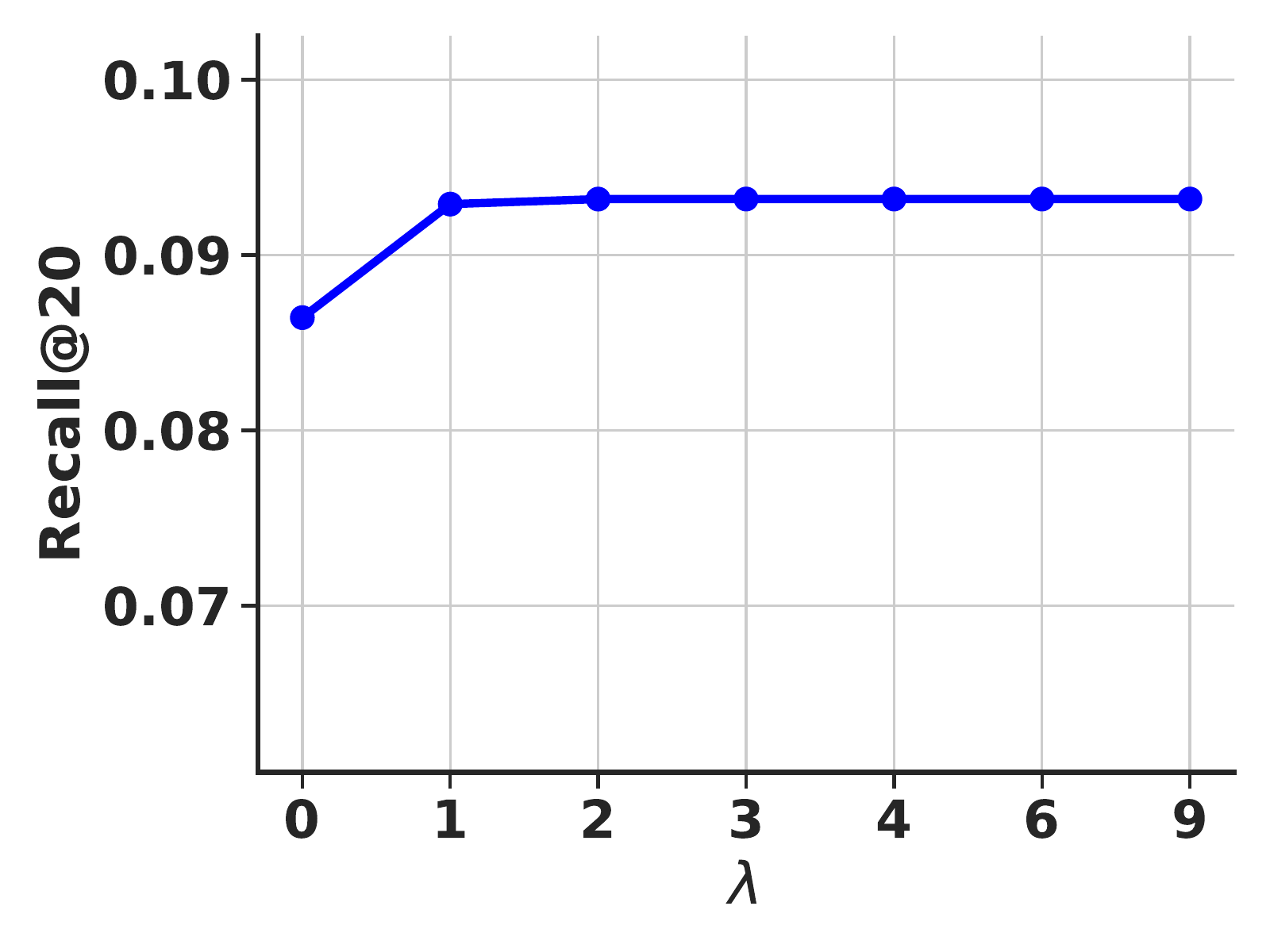}}}}
{\subfigure[Gowalla - NDCG@20]
{\includegraphics[width=0.22\linewidth]{{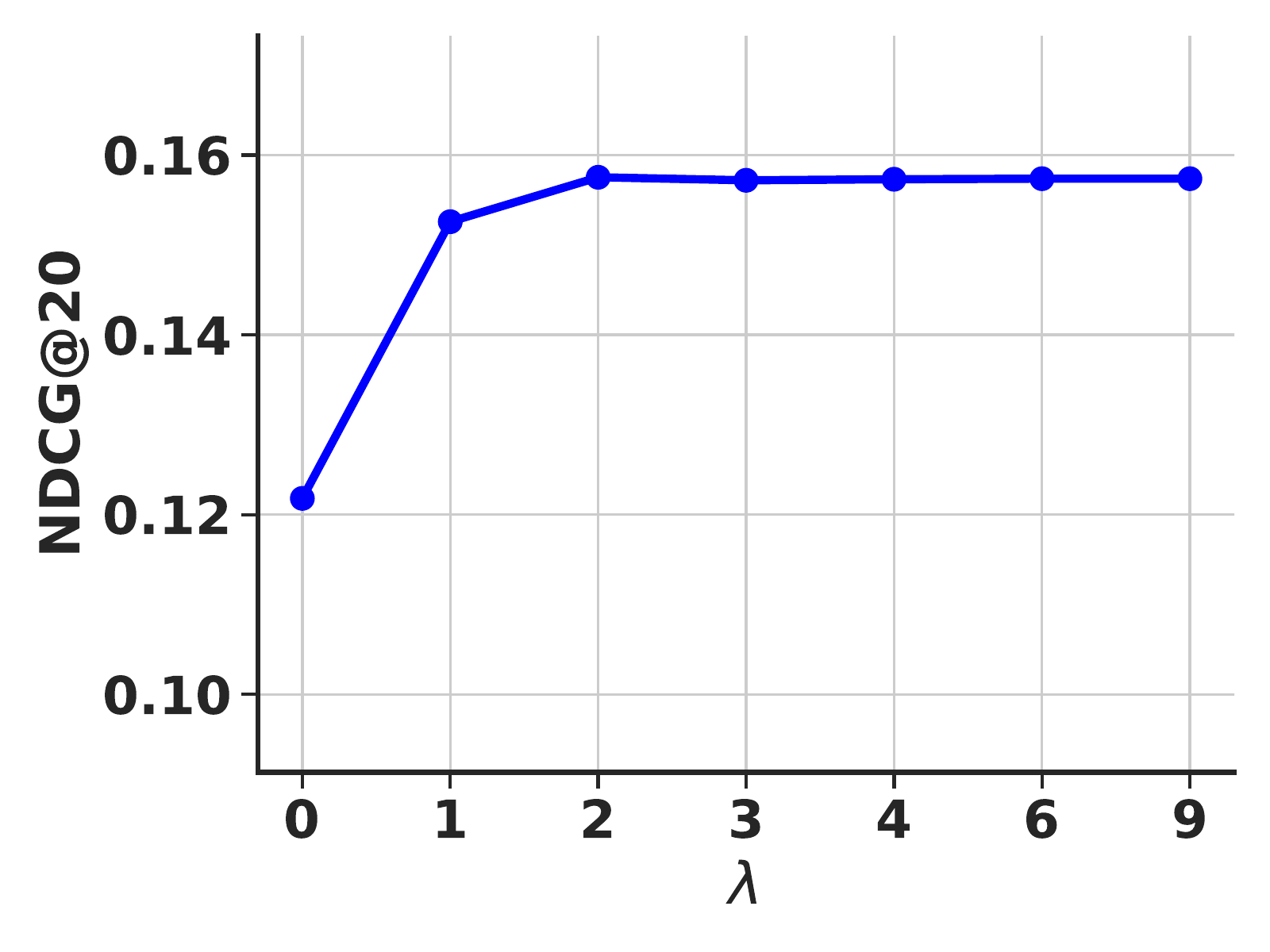}}}}
{\subfigure[Yelp2018 - NDCG@20]
{\includegraphics[width=0.22\linewidth]{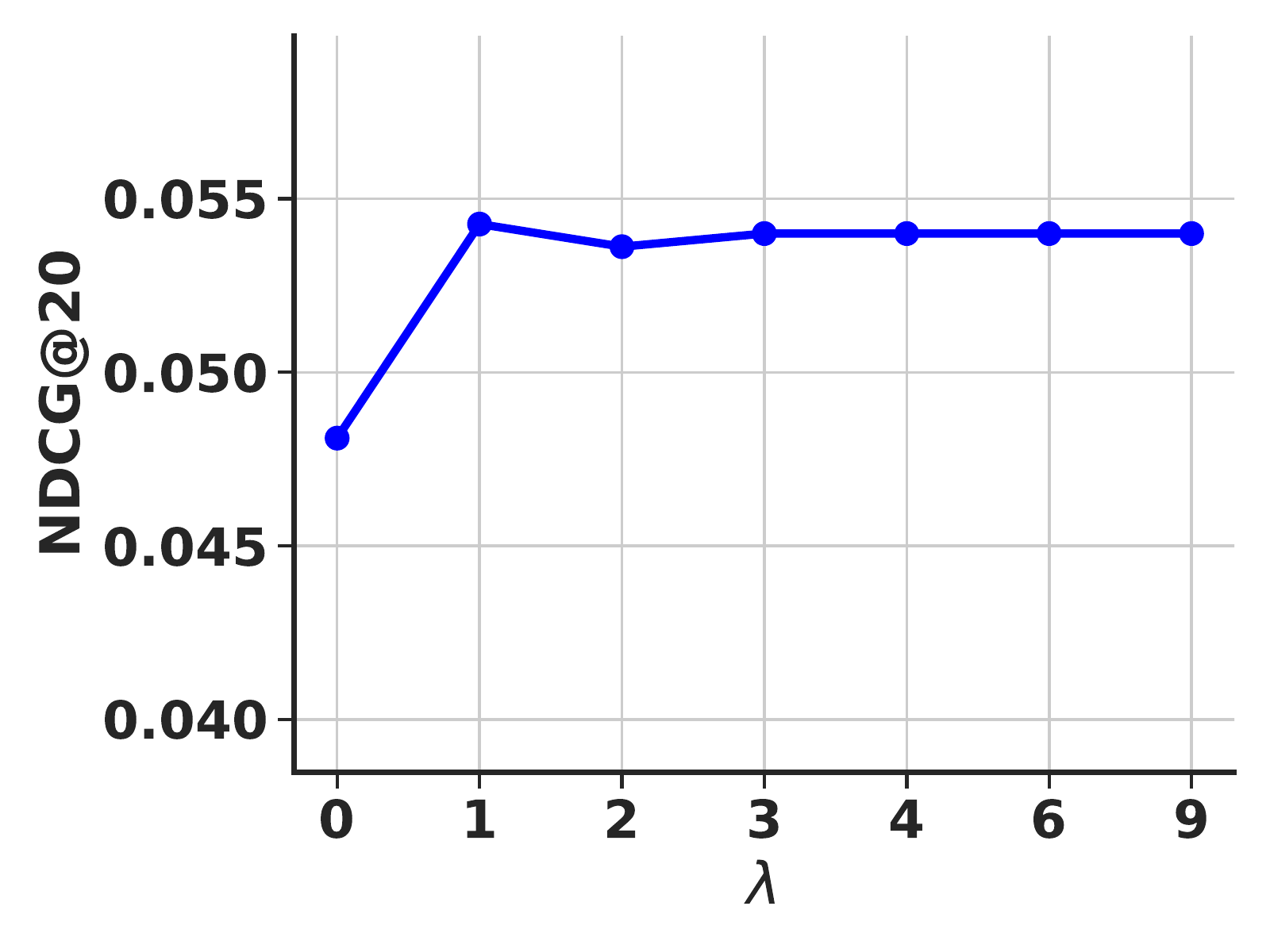}}}
{\subfigure[Amazon-book - NDCG@20]
{\includegraphics[width=0.22\linewidth]{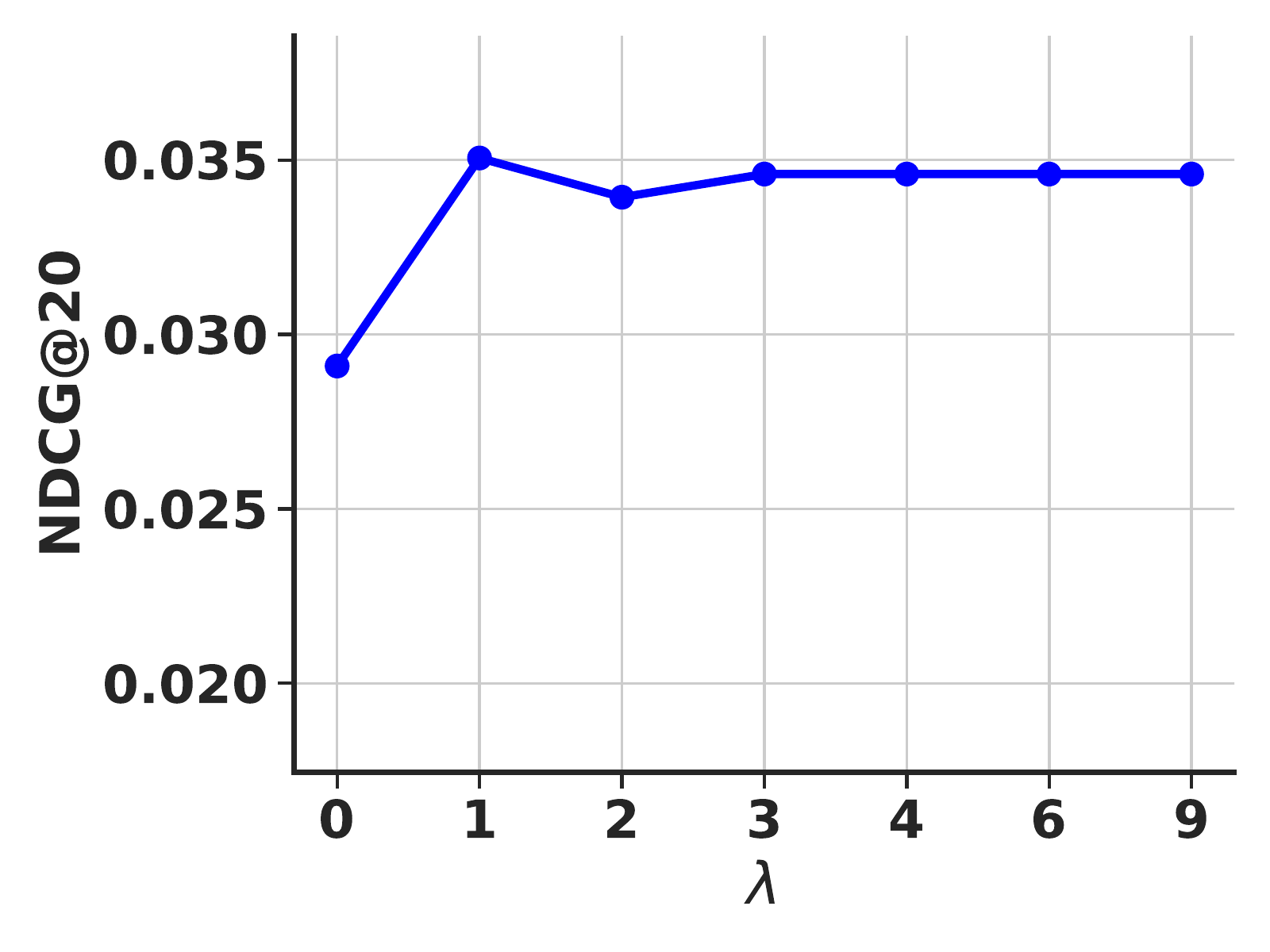}}}
{\subfigure[LastFM - NDCG@20]
{\includegraphics[width=0.22\linewidth]{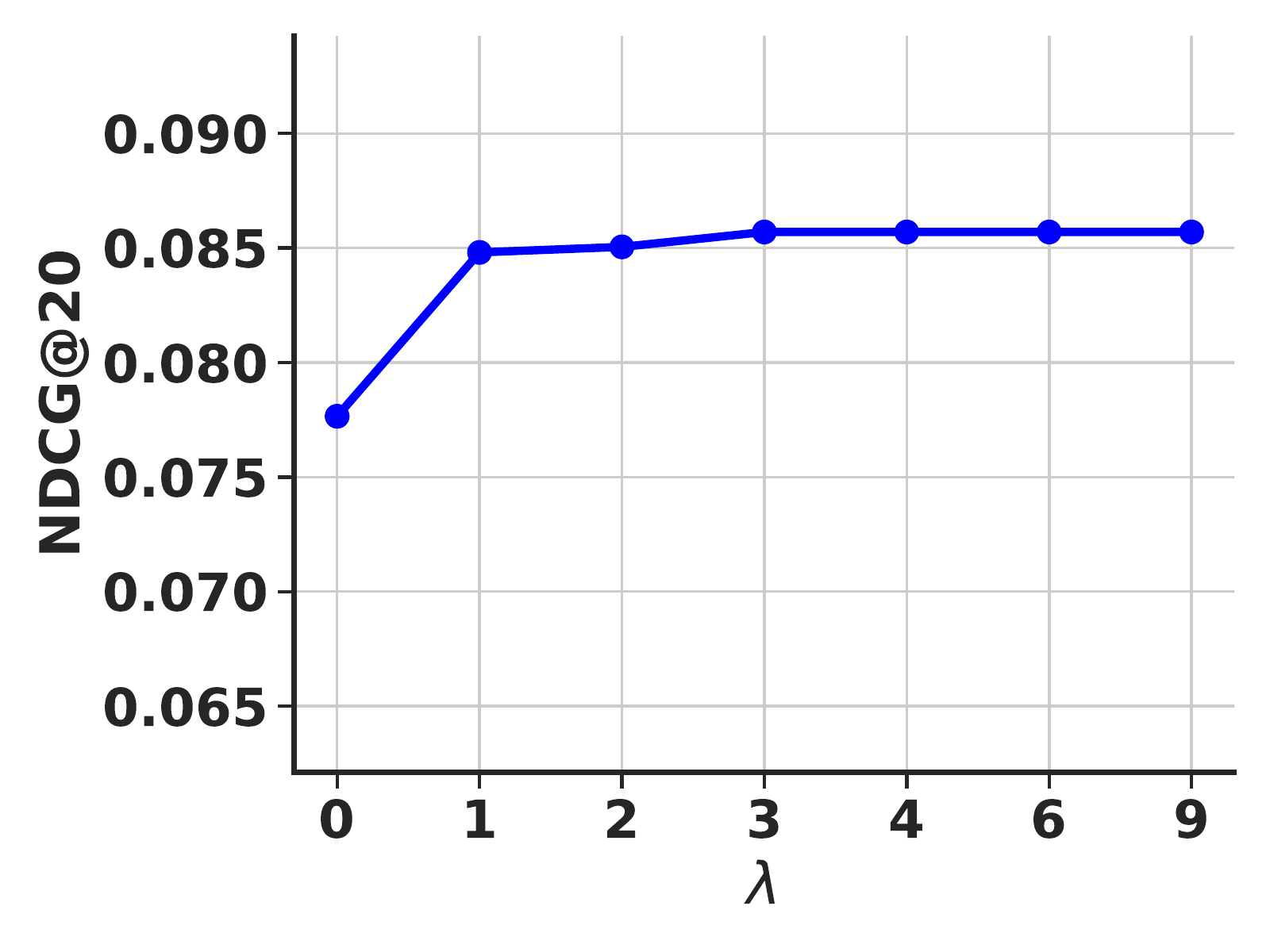}}}
\vskip -0.2in
\caption{The effect of hyper-parameter $\lambda$ under Recall@20 and NDCG@20 metrics.}\label{fig:lambda}
\end{figure*}

\vskip -0.1in

\subsubsection{\textbf{Sparsity pattern}}

To further investigate the piecewise constant trend captured by the proposed graph trend collaborative filtering, we study the sparsity patterns in the user-item embedding differences over the interactions, $\tDelta \vE^K$. We define the sparsity ratio as the portion of elements whose absolute value is less than 0.2. The results are summarized in Table~\ref{tab:sparse}. It can be observed that for all datasets, the proposed GTN has a much larger portion of sparse elements in the embedding differences than others. This verifies the trend captured by the proposed GTCF. 

\begin{table}[!htbp]
\centering
\vskip -0.10in
 \caption{Sparsity ratio in the user-item embedding differences.}
\vskip -0.150in
\label{tab:sparse}
\scalebox{0.85}
{
\begin{tabular}{|c|c|c|c|c|c|}
\hline
\multicolumn{2}{|c|}{\textbf{Datasets}}              & \textbf{Gowalla} & \textbf{Yelp2018} & \textbf{Amazon-Book} & \textbf{LastFM} \\ \hline \hline
\multirow{3}{*}{\rotatebox{90}{\textbf{Method}}} & \textbf{NGCF}     & 0.045            & 0.0289            & 0.0403               & 0.0835          \\ \cline{2-6} 
                                 & \textbf{LightGCN} & 0.017            & 0.017             & 0.0233               & 0.0218          \\ \cline{2-6} 
                                 & \textbf{\ourname{} (Ours)}      & 0.1241           & 0.1424            & 0.2612               & 0.2282          \\ \hline
\end{tabular}
}
\vskip -0.20in

\end{table}


\subsection{\textbf{Ablation Study of \ourname{}}}


 \subsubsection{\textbf{Effect of $\lambda$}}
We first investigate the impact of the hyper-parameter $\lambda$ in \ourname{}. The value of $\lambda$  controls the impact of the embedding smoothness. Figure~\ref{fig:lambda} shows the performance change of \ourname{} w.r.t. Recall@20 and NDCG@20. The value of $\lambda$ is searched from 0 to 9. We can find that \ourname{} is relatively insensitive to $\lambda$. More specifically, when $\lambda$ is larger than 2, the recommendation performance keeps stable. 
Moreover, increasing the value of $\lambda$ from 0 to 1 significantly improves the recommendation performance, indicating that the smoothness captured by the proposed graph trend collaborative filtering is helpful. 
 
\subsubsection{\textbf{Effect of Number of Layers}}

To study whether our proposed method \ourname{} can benefit from stacking more propagation layers, we vary the number of layers $k$ in the range of $\{1, 2, 3, 4, 5\}$ and report the performance of NGCF, LightGCN and \ourname{} on all datasets in  Table~\ref{tab:num_layers}. From the table, we make the following observations:  

\begin{itemize}[leftmargin=*] 
\item \vspace{-0.05in}
With the increase of model depth,  the recommendation performances of all GNN-based methods first increase and then decrease. 
Clearly, in most cases, the GNN-based methods achieve the peak performance with a model depth of 3, while more layers can degrade the prediction performance. This finding is consistent with the general observation of existing GNN-based methods such as NGCF and LightGCN. 

\item  Comparing with other methods (i.e., NGCF and LightGCN),  we can observe that \ourname{}  achieves the best prediction performance on  four datasets w.r.t. Recall@20 and NDCG@20 in most cases. More specifically, with five layers, our proposed method \ourname{} achieves significant improvement over the strongest baselines by 19\% on Recall@20 and 23\% on NDCG@20 in LastFM dataset. These general observations justify the effectiveness of adaptive smoothness captured by graph trend collaborative filtering. 
 
\end{itemize}

\subsubsection{\textbf{Convergence of GTCF}}
In Figure~\ref{fig:loss}, the loss value of the objective in Eq.~\eqref{eq:gtf_obj} is showed during the filtering process. Those four curves for different datasets verify the convergence of the proposed graph trend collaborative filtering.

\begin{figure*}[htbp]
 \centering
 \vskip -0.10in
{\subfigure[Gowalla]
{\includegraphics[width=0.22\linewidth]{{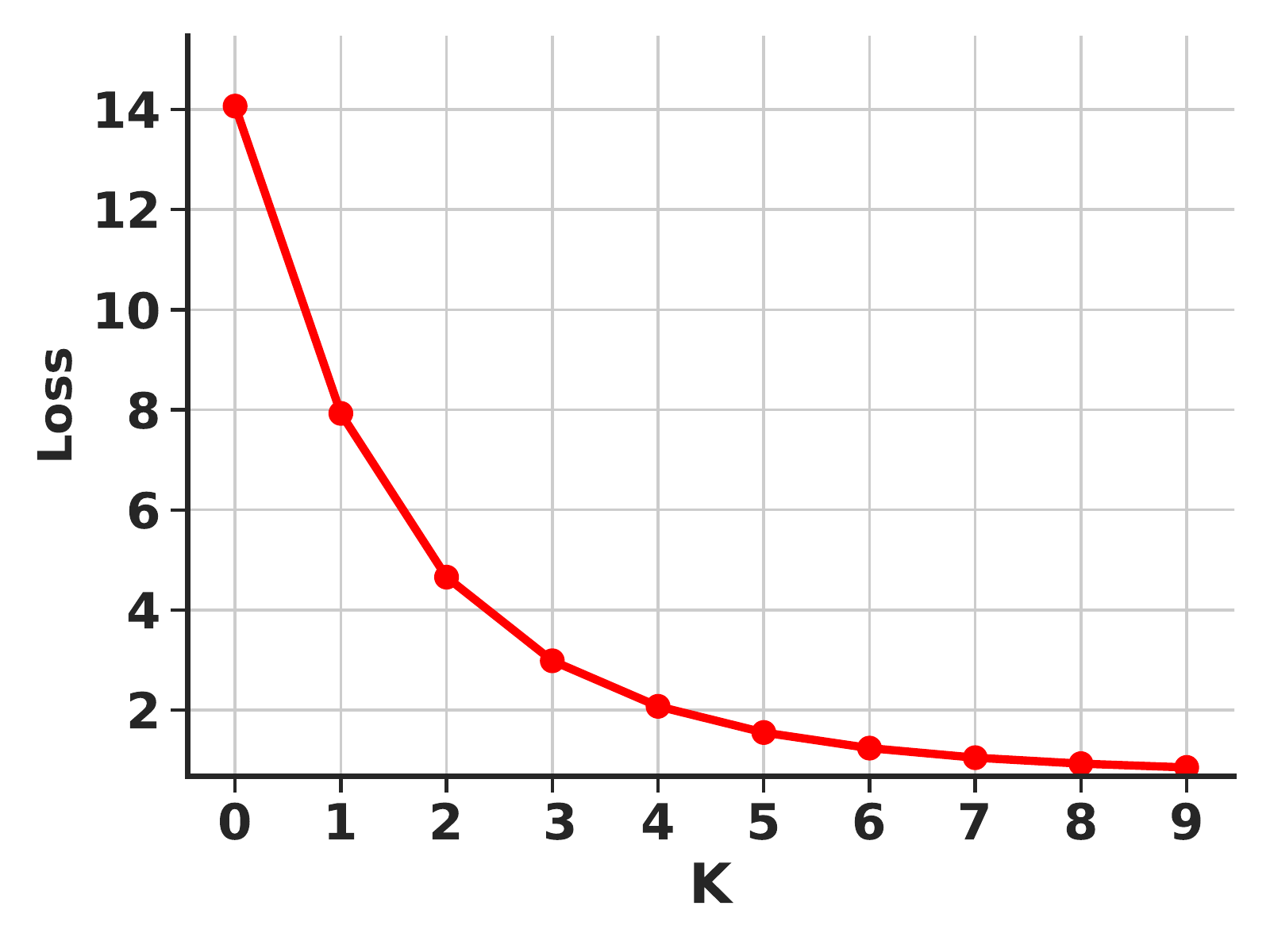}}}}
{\subfigure[Yelp2018]
{\includegraphics[width=0.22\linewidth]{{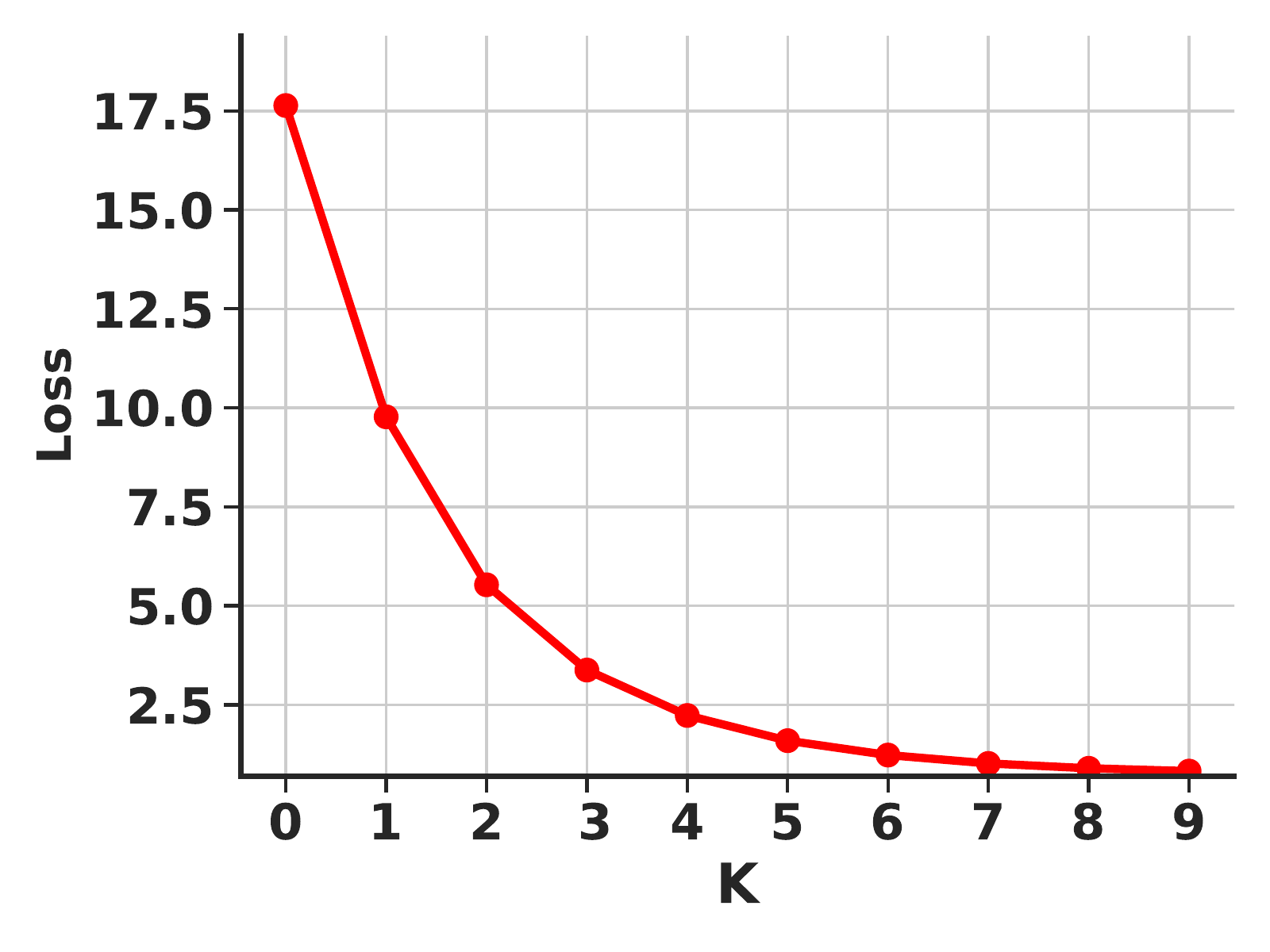}}}}
{\subfigure[Amazon-book]
{\includegraphics[width=0.22\linewidth]{{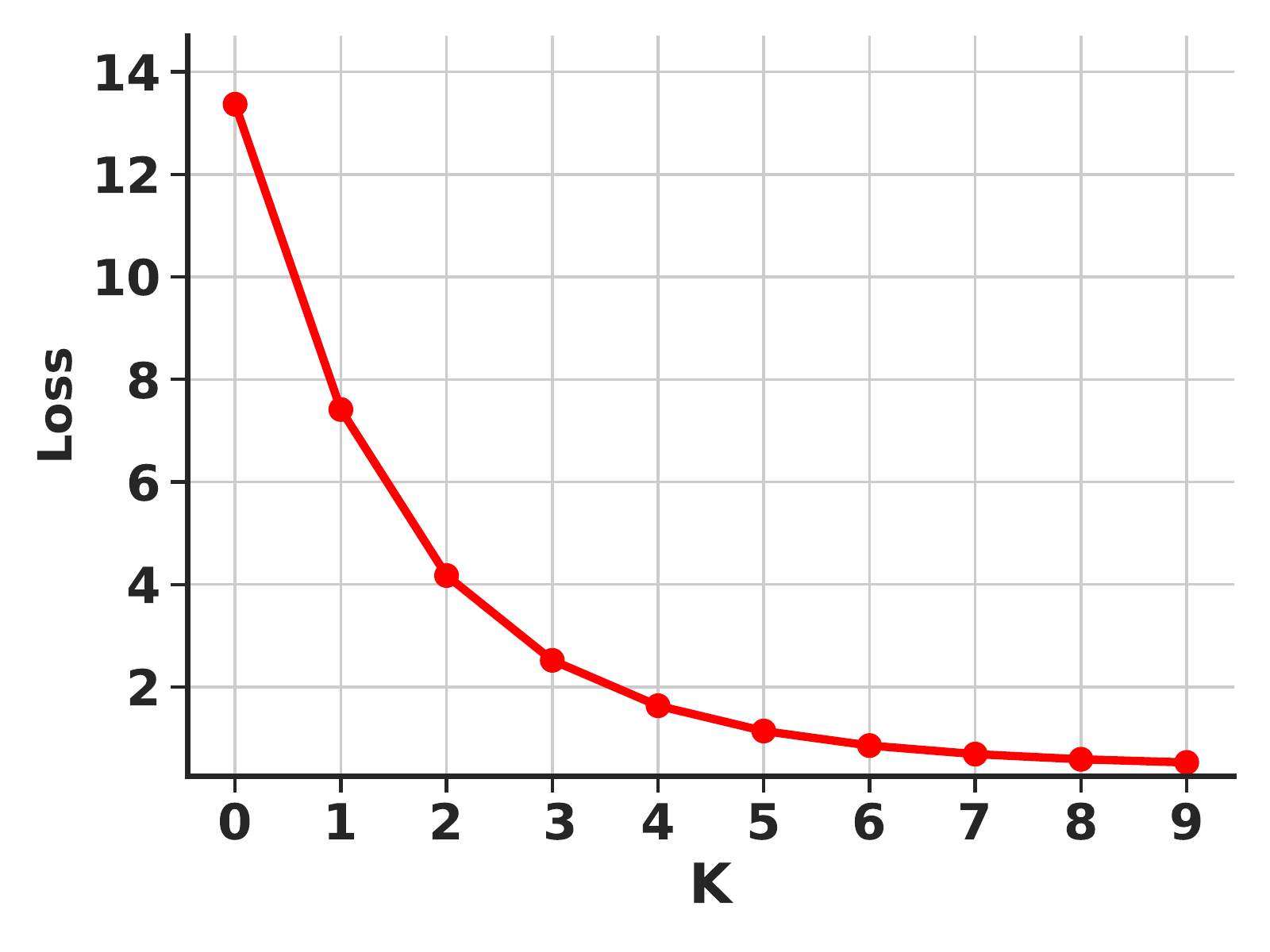}}}}
{\subfigure[LastFM]
{\includegraphics[width=0.22\linewidth]{{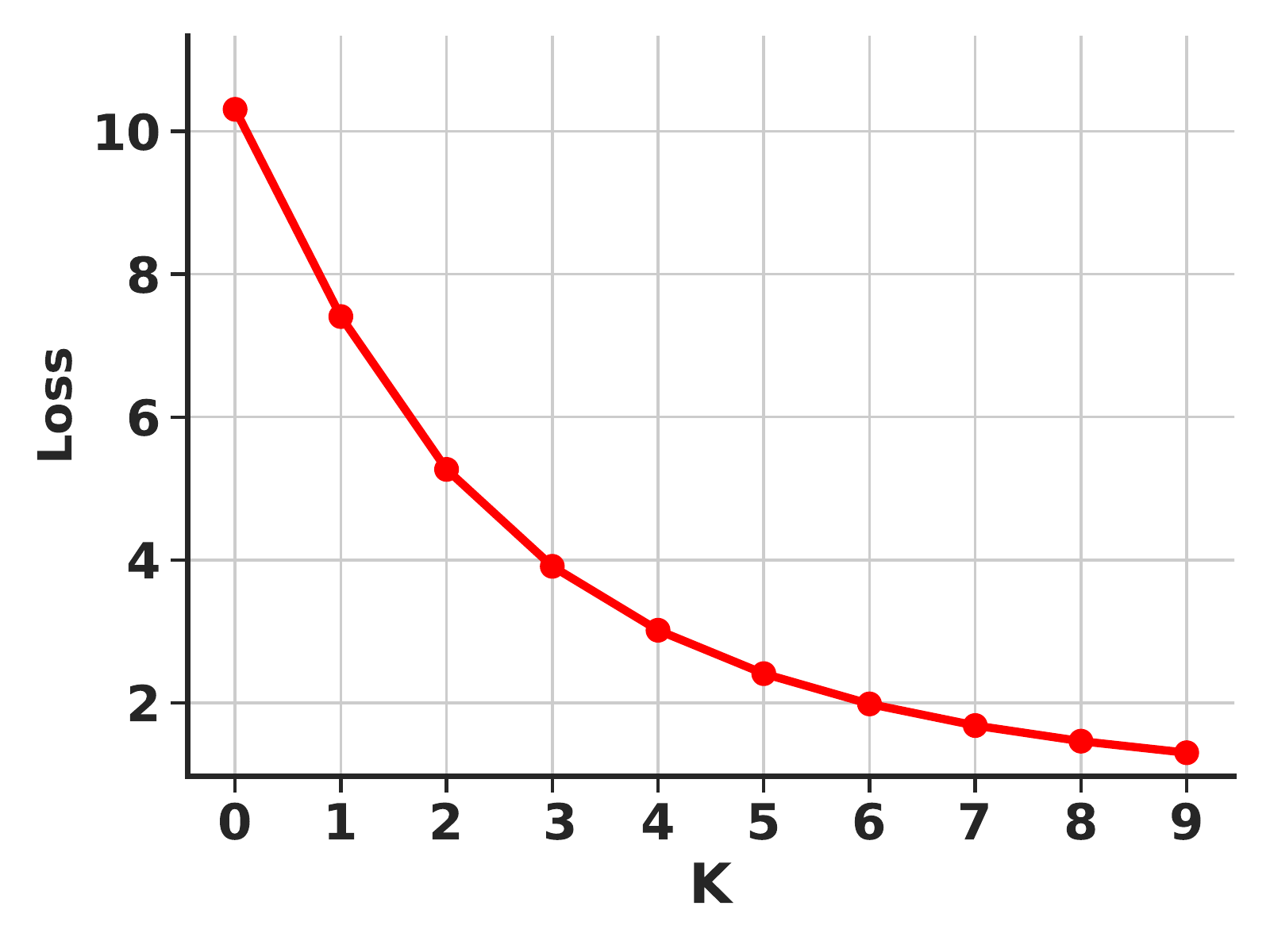}}}}
\vskip -0.20in
 \caption{Loss value of Eq.~\eqref{eq:gtf_obj} in the graph trend collaborative filtering process.}\label{fig:loss}
 \vskip -0.150in
 \end{figure*}

\subsubsection{\textbf{Effect of Training Epochs}}  We plot the performance curves of NGCF, LightGCN and \ourname{} w.r.t. training epochs in Figure~\ref{fig:training_curves}. From the figure, we can observe that at the same epoch {\ourname} always outperforms the other two methods and it also converges faster than the other two methods in most cases.

\begin{figure*}[htbp]
\centering
{\subfigure[Gowalla - Recall@20]
{\includegraphics[width=0.2\linewidth]{{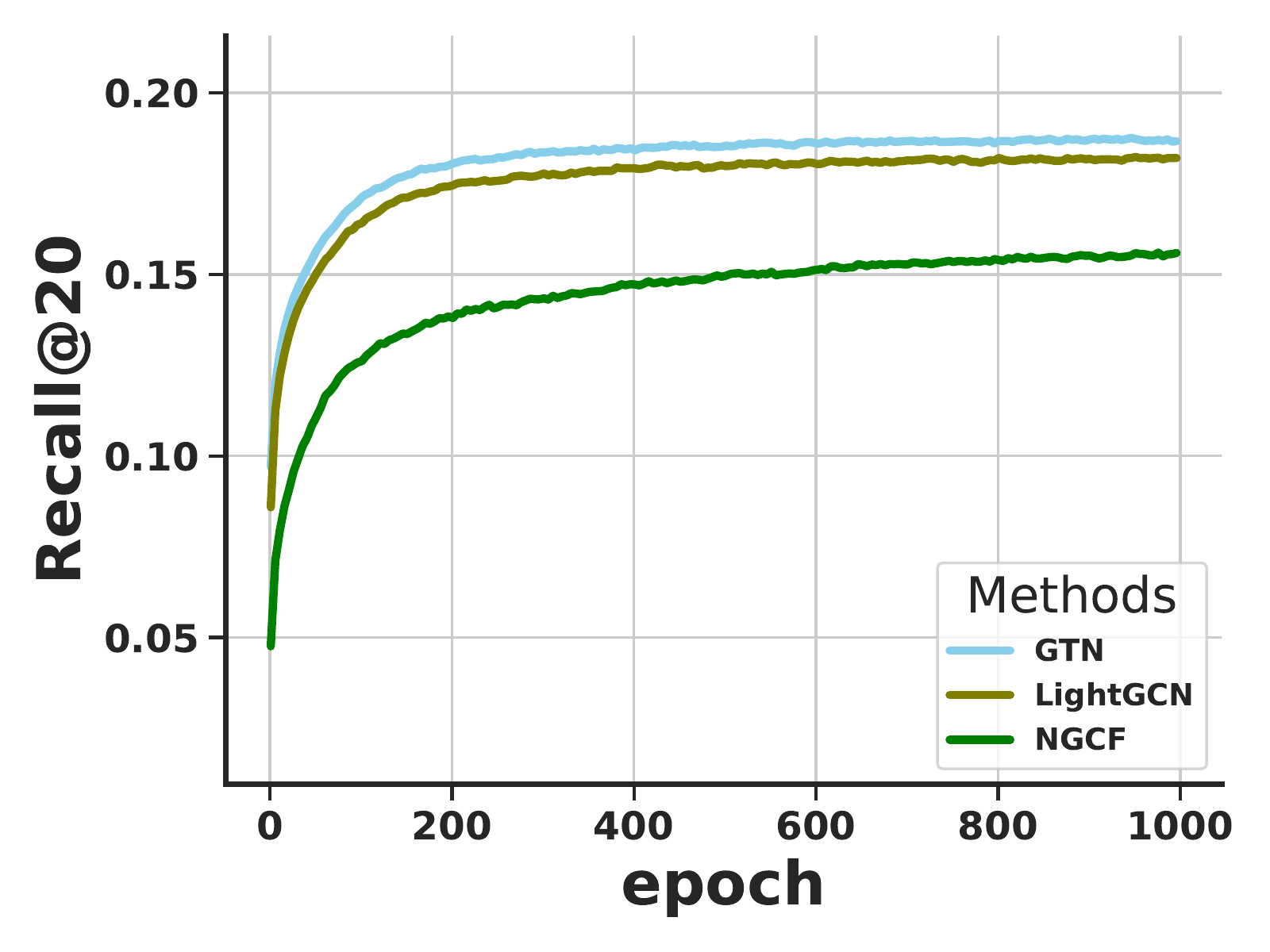}}}}
{\subfigure[Yelp2018 - Recall@20]
{\includegraphics[width=0.2\linewidth]{{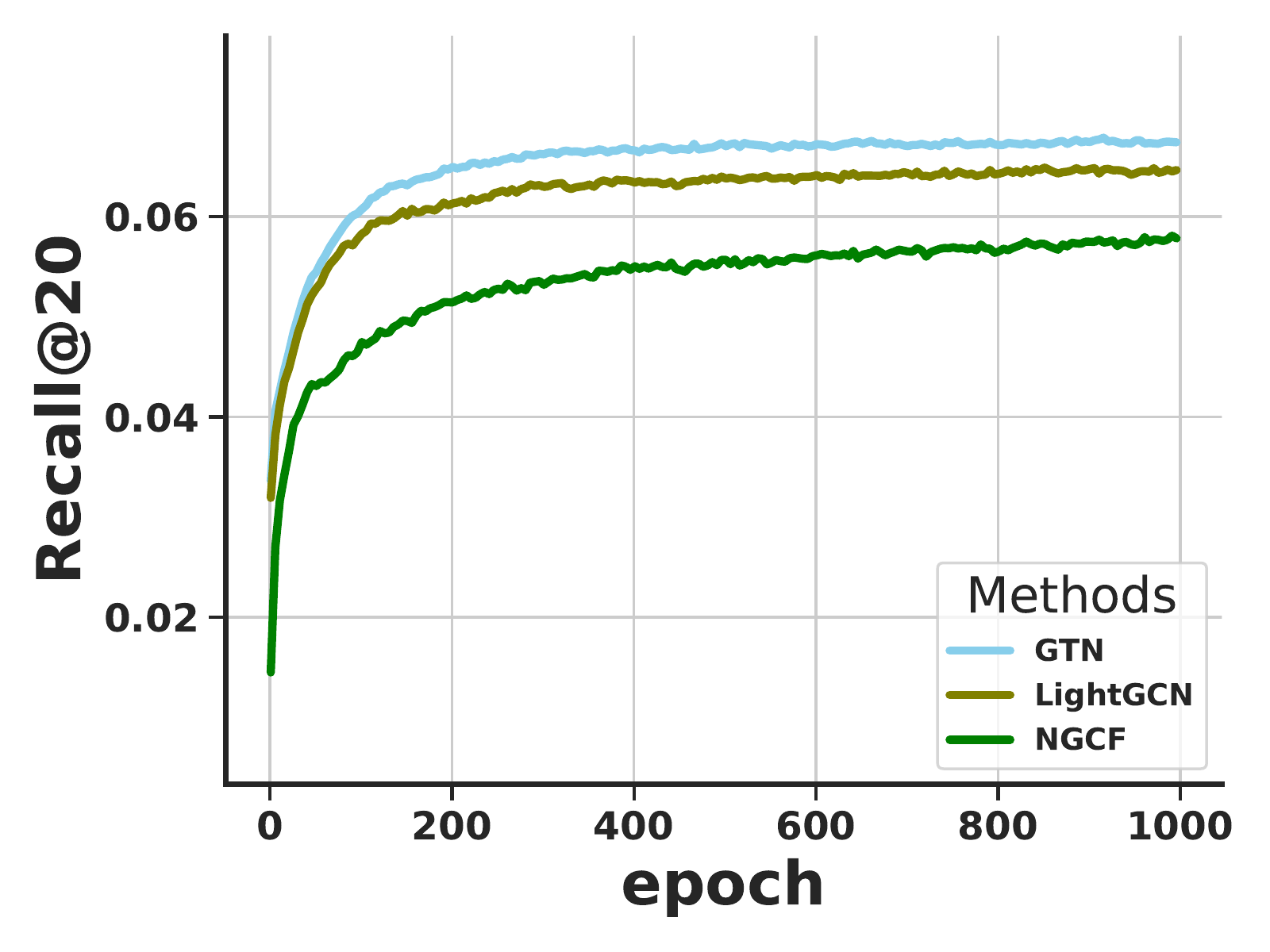}}}}
{\subfigure[Amazon-book - Recall@20]
{\includegraphics[width=0.2\linewidth]{{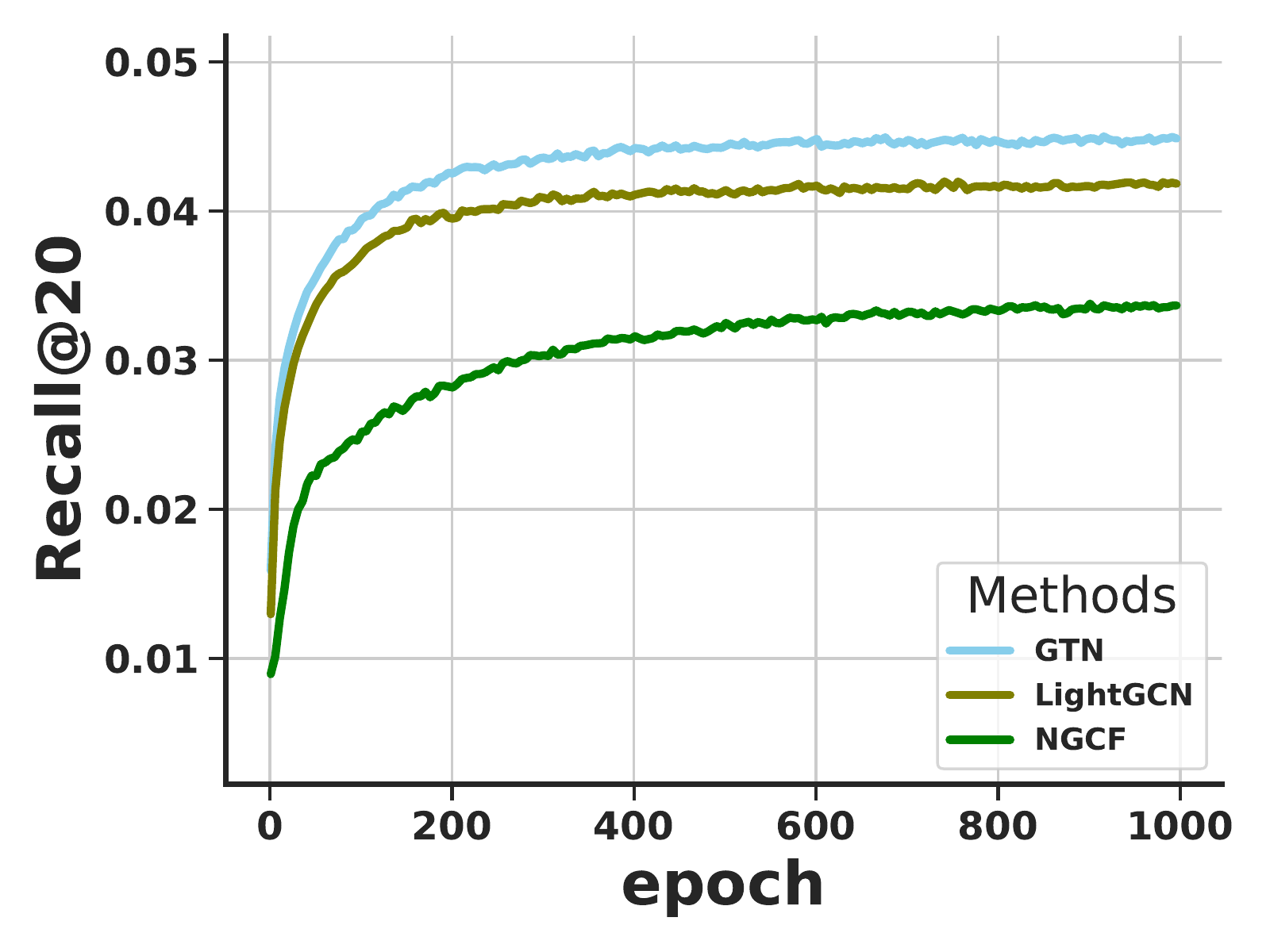}}}}
{\subfigure[LastFM - Recall@20]
{\includegraphics[width=0.2\linewidth]{{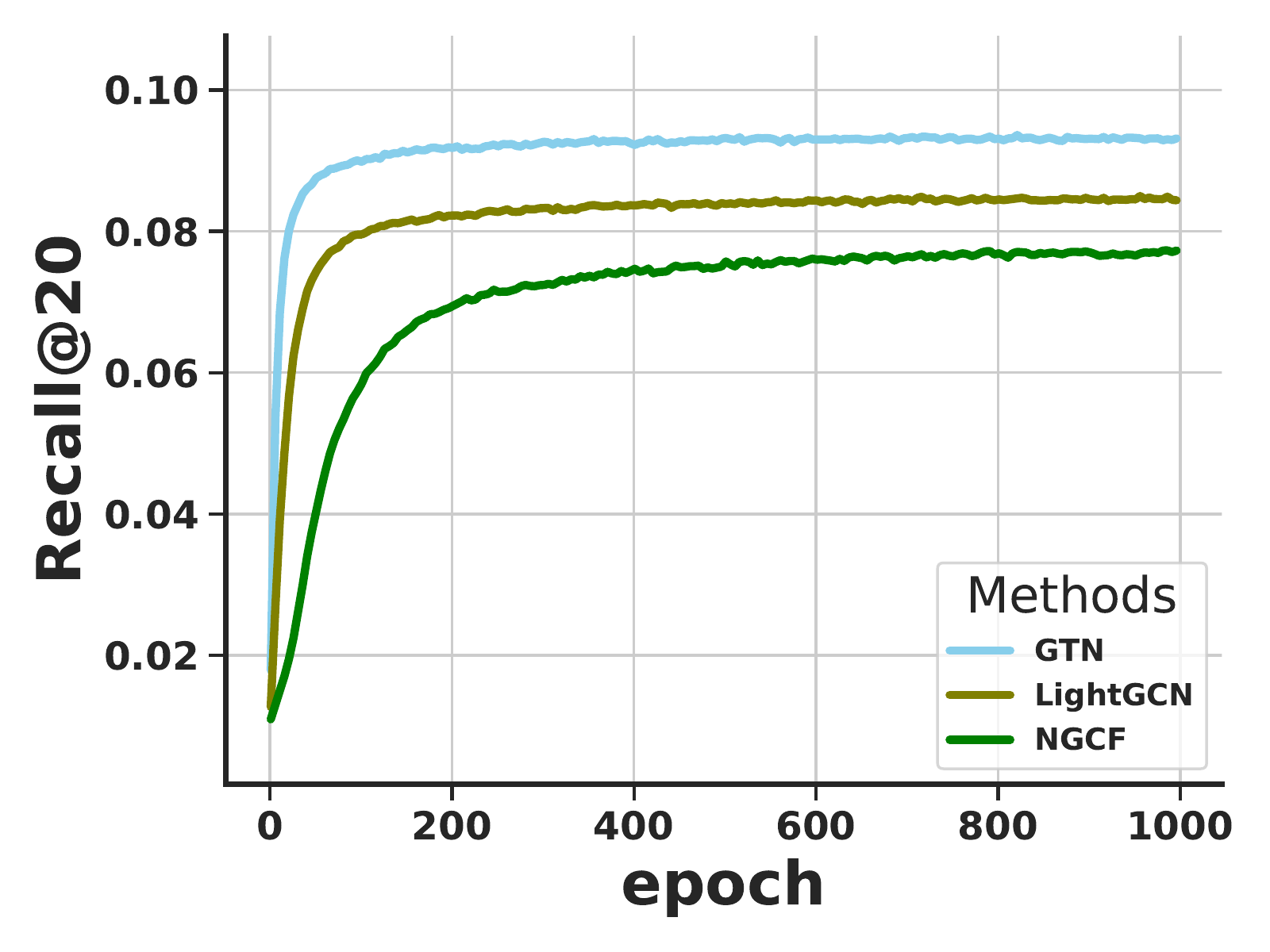}}}}

{\subfigure[Gowalla - NDCG@20]
{\includegraphics[width=0.2\linewidth]{{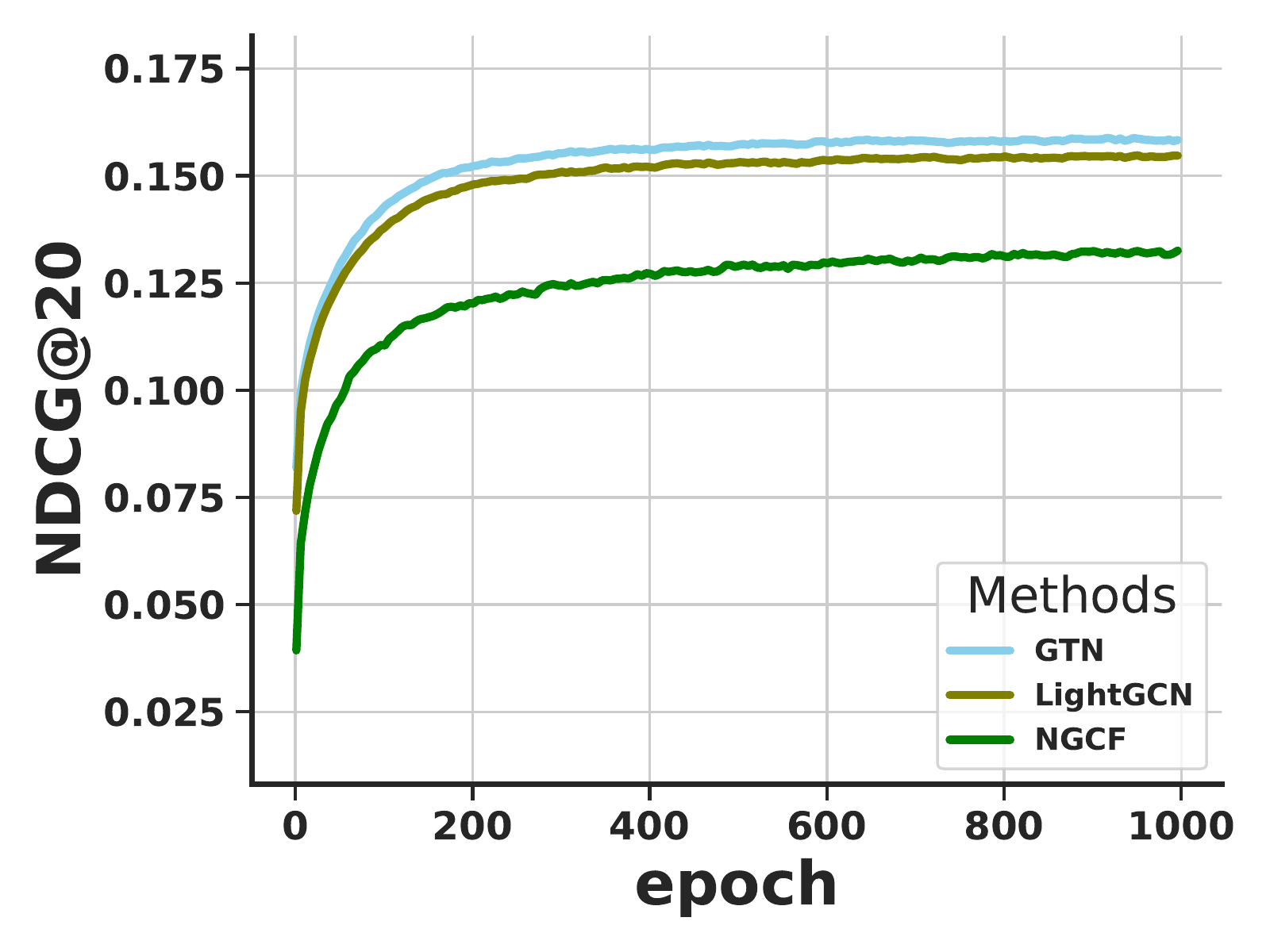}}}}
{\subfigure[Yelp2018 - NDCG@20]
{\includegraphics[width=0.2\linewidth]{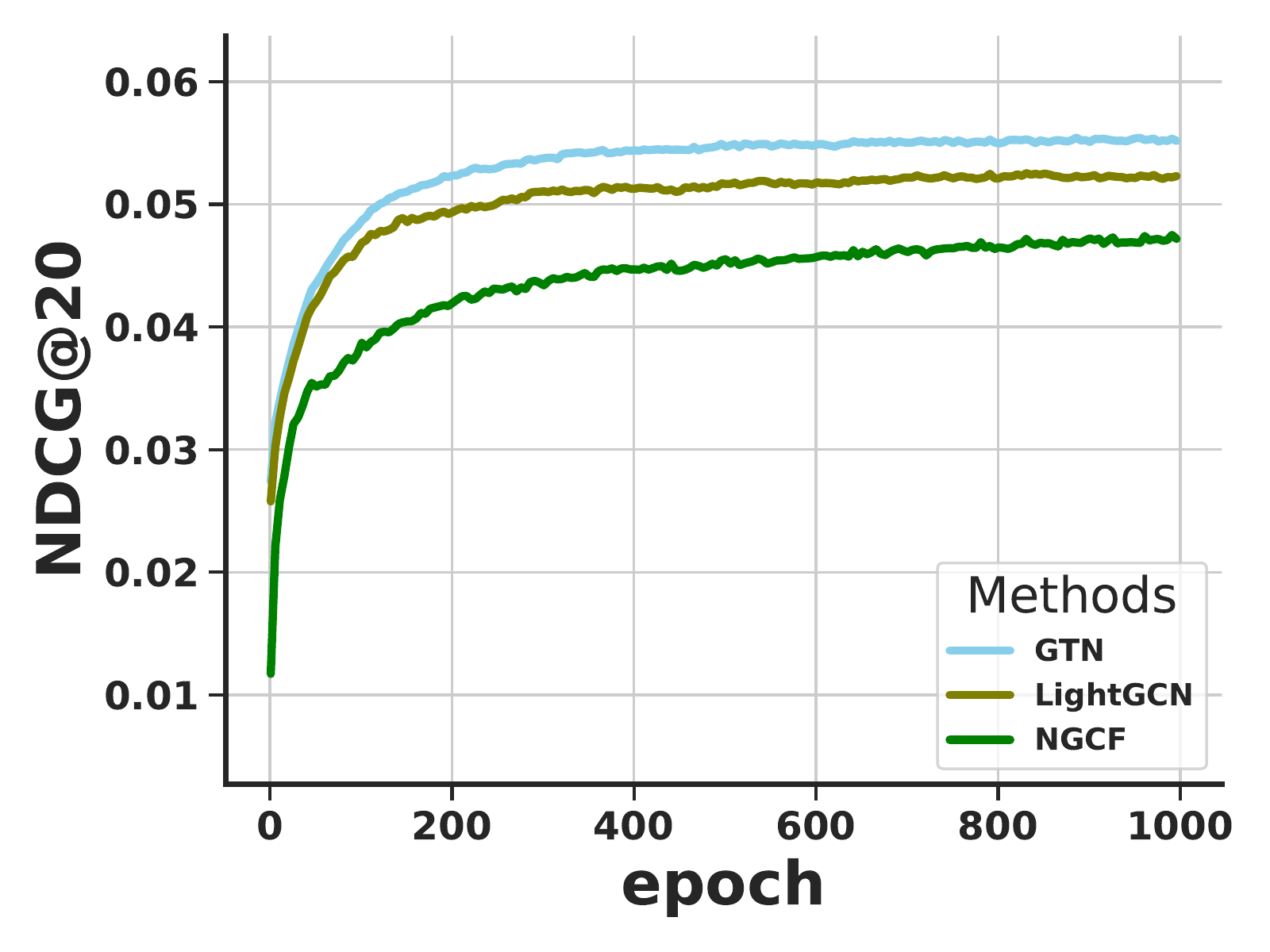}}}
{\subfigure[Amazon-book - NDCG@20]
{\includegraphics[width=0.2\linewidth]{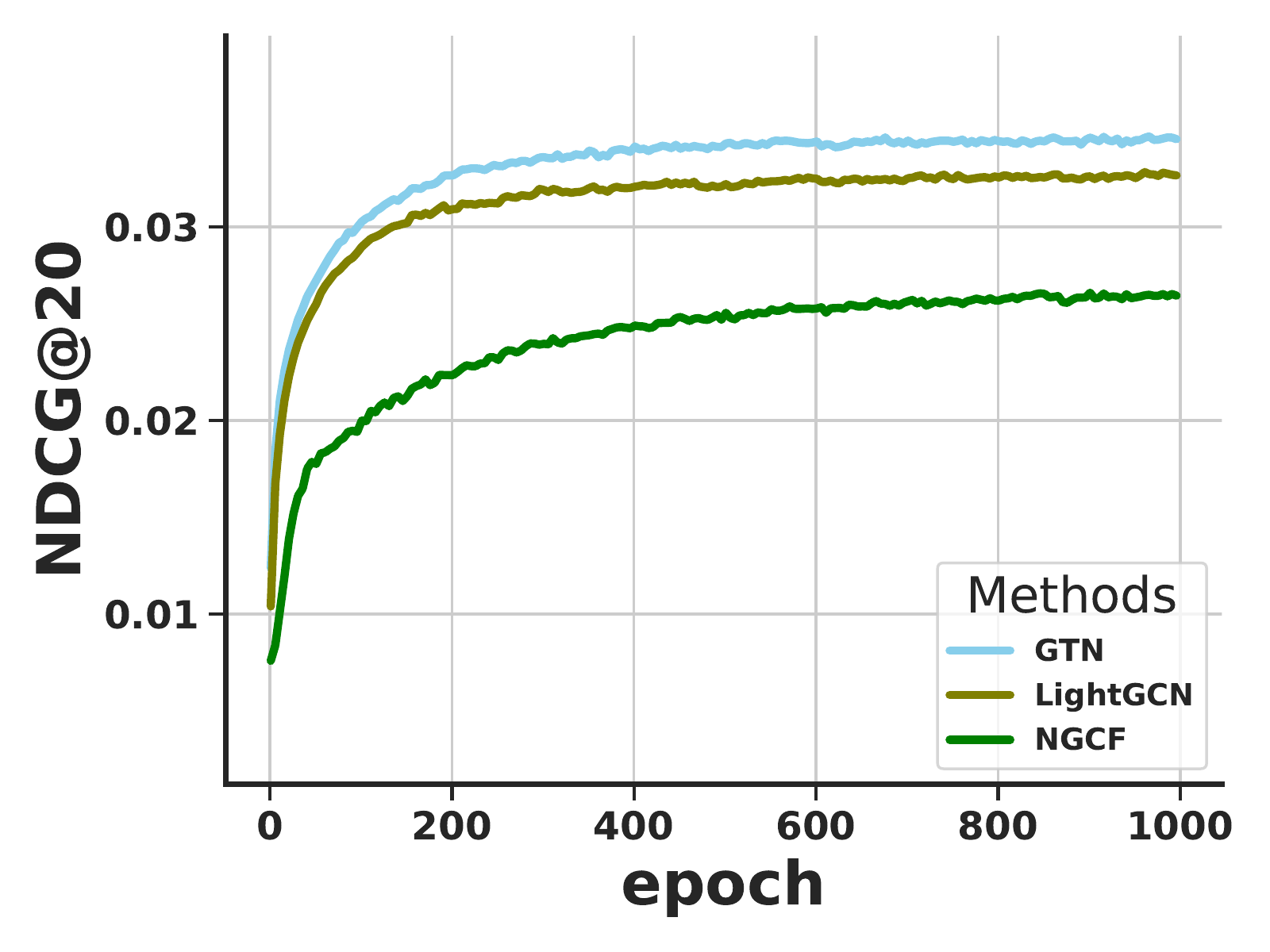}}}
{\subfigure[LastFM - NDCG@20]
{\includegraphics[width=0.2\linewidth]{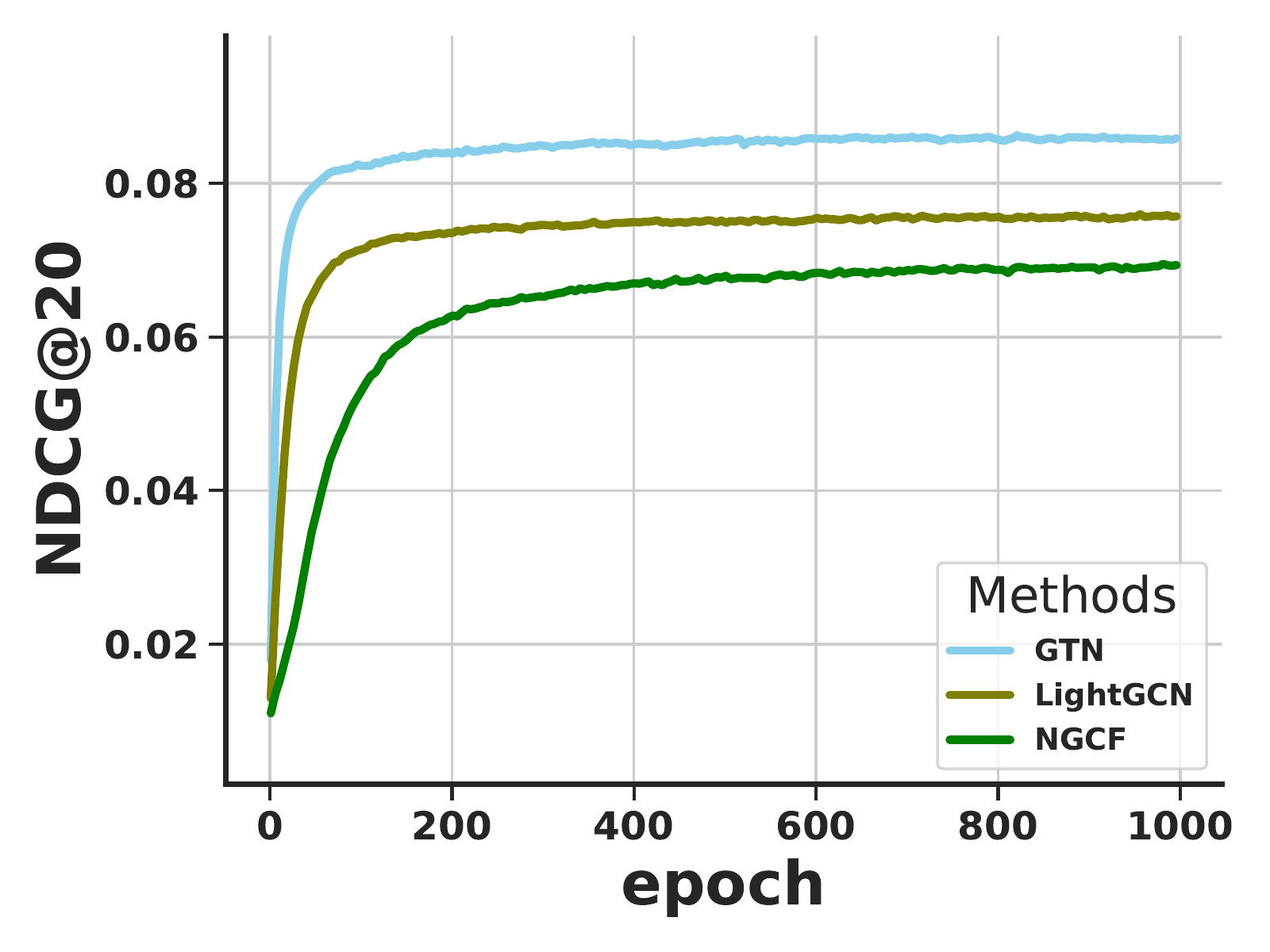}}}
\vskip -0.2in
\caption{Performance curves of NGCF, LightGCN, and \ourname{} during training.}\label{fig:training_curves}
\vskip -0.15in
\end{figure*}
\section{Related Work}
\label{sec:relatedwork}


\subsection{Collaborative Filtering}
Collaborative Filtering (CF) is one of the most popular techniques in the modern recommender systems~\cite{wang2015collaborative,zhao2018recommendations}. It assumes that users  who behave similarly are likely to have similar preferences towards items~\cite{tang2013social}. To capture users' preferences towards items, one widely used paradigm of CF technique is  to decompose user-item interaction data into embedding vectors of users and items~\cite{fan2018deep}.
For instance,  among the various collaborative filtering techniques, matrix factorization (MF) is the most popular one, which aims to learn representation vectors to represent users and items and performs inner product between them to make predictions~\cite{Rendle2009BPRBP,xue2017deep}. 
Later on, with the great success of deep neural networks on non-linear representation learning~\cite{fan2020deep,liu2020does,wang2020global,zhao2022multi},  recent efforts have been made to apply deep neural networks to recommendation tasks and they have shown promising prediction performance. For example, NeuCF~\cite{He2017NeuralCF} is proposed to replace the inner product of MF model with nonlinear neural networks to model non-linear interactions between users and items. DSCF~\cite{fan2019deep_dscf} adopts  a bi-directional long short-term memory network (Bi-LSTM) based deep language model to implicitly capture information from high-order neighbors to enhance user representations for recommendations.

\subsection{GNNs-based Recommender Systems}

Despite the great success achieved by the aforementioned CF methods, most of them  treat each interaction as an independent instance and fail to capture deep collaborative signals between users and items, easily leading to  sub-optimal representations of users and items~\cite{Wang2019NeuralGC}.
Recent years have witnessed the great success of  graph neural networks (GNNs) techniques in representation learning for graph data~\cite{ma2020deep_book,wu2020comprehensive,jin2021automated,jin2020graph_robust,derr2018signed,jin2021node}. The main idea of GNNs adopts a message-passing scheme to learn node embedding from local neighbors via  aggregation and transformation operations. Meanwhile,  data in recommender systems can be naturally denoted as graphs, where users and items can be denoted as nodes and their implicit interactions can be denoted as edges in the graph~\cite{fan2020graph}. Hence, it is desirable to build recommender systems based on graph neural networks to capture deep collaborative signals in the user-item graph. For example, as the extension of a famous GNN variant GraphSage~\cite{hamilton2017inductive} in recommender systems, PinSage~\cite{ying2018graph} is applied to learn the embedding of nodes in web-scale graphs for item recommendation. GraphRec~\cite{fan2019graph,fan2020graph}  introduces a graph attention network framework to encoder user-item interactions and user-user social relations for social recommendations. NGCF~\cite{Wang2019NeuralGC} proposes to explicitly encode collaborative signals in the form of high-order connectivities by message-passing propagation for  recommendations. Later on,  LightGCN~\cite{He2020LightGCNSA} is introduced to largely simplify the NGCF model by removing feature transformation and nonlinear activation, achieving the state-of-the-art prediction performance for recommendations.

\subsection{Trend Filtering}
$\ell_1$ trend filtering~\cite{kim2009ell_1} was first proposed in 2009 to estimate the trend in time series data. Later on, it was studied for nonparametric regression to provide adaptive piecewise polynomial estimation which adapts to the local level of smoothness~\cite{liu2021elastic,tibshirani2014adaptive, wakayama2021locally,chen2016representations}. 
Ortelli et.al~\cite{ortelli2021tensor} proposed to extend the trend filtering for tensor denoising.
The idea of trend filtering was also extended to the general graph signal by generalizing the difference operator on 1D grid to the graph difference operator on general graphs~\cite{wang2016trend,chen2014signal}. It also exhibited piecewise polynomial behaviors on graphs.
The work of~\cite{chen2016representations} proposed multi-resolution local sets based piecewise-constant and piecewise-smooth dictionaries as graph signal representations on graphs. 
Recently, graph trend filtering also inspires the design of graph neural networks such as Elastic GNNs~\cite{liu2021elastic}, which have been shown to be much more robust against adversarial graph attacks due to the enhanced local smoothness. 
To the best of our knowledge, trend filtering has not been studied in recommender systems. In this work, we make the first attempt to investigate trend filtering as a new collaborative filtering technique for recommender systems, and we hope the novel and principled design of graph trend filtering networks (GTN) can provide new inspirations in this important field.

\section{Conclusion}
\label{sec:conclusion}

Although existing graph neural networks based recommender systems achieve promising prediction performance, they are unable to discover the heterogeneous reliability of interactions among instances during the embedding propagation. 
In this work, we first analyze the drawbacks of the existing state-of-art GNNs-based collaborative filtering methods, such as non-adaptive propagation and non-robustness through the perspective of Laplacian smoothing in the user-item graph. To address these drawbacks,  we introduce a principled graph trend collaborative filtering technique and propose a novel graph trend filtering networks framework (\ourname{}) to capture the adaptive reliability of the interactions between users and items. 
Extensive experiments on four real-world datasets demonstrate the effectiveness of our proposed method. Moreover, the empirical study also demonstrates the benefits of modeling the adaptive reliability of the interactions between users and items for recommendations.

\begin{acks}
The research described in this paper has been partly supported by NSFC (project no. 62102335), an internal research fund from The Hong Kong Polytechnic University (project no. P0036200), a General Research Fund from the Hong Kong Research Grants Council (Project No.: PolyU 15200021). 
Xiangyu Zhao is supported by Start-up Grant (No.9610565) for the New Faculty of the City University of Hong Kong and the CCF-Tencent Open Fund. 
Xiaorui Liu, Wei Jin, and Jiliang Tang are supported by the National Science Foundation (NSF) under grant numbers IIS1714741, CNS1815636, IIS1845081, IIS1907704, IIS1928278, IIS1955285, IOS2107215, and IOS2035472, the Army Research Office (ARO) under grant number W911NF-21-1-0198, the Home Depot, Cisco Systems Inc and SNAP.

\end{acks}

\newpage
\balance
\bibliographystyle{ACM-Reference-Format}
\bibliography{references/references}

\end{document}